\documentclass[11pt,fleqn]{article}
\usepackage{etex}
\usepackage{amssymb,latexsym,amsmath,amsfonts,graphicx, ebezier}
\usepackage{pictex,color}
\usepackage{pict2e}
\usepackage[textsize=footnotesize]{todonotes}
    \advance\textheight by \topskip

    \numberwithin{equation}{section}

\makeatletter
\def\eqalign#1{\null\vcenter{\def\\{\cr}\openup\jot\m@th
  \ialign{\strut$\displaystyle{##}$\hfil&$\displaystyle{{}##}$\hfil
      \crcr#1\crcr}}\,}
\makeatother
\usepackage{lmodern} 

\usepackage{graphicx} 

\usepackage{amsmath, accents}
\usepackage{amssymb,tikz}
\usepackage{pict2e}
\usepackage{amsthm}
\usepackage{amscd}
\usepackage{mathrsfs}
\usepackage{todonotes}

\usepackage{yfonts}

\usepackage{marvosym}

\def\beq{\begin{equation}}
\def\eeq{\end{equation}}

\newcommand{\be}{\begin{equation}}
\newcommand{\ee}{\end{equation}}

    \def\ds{\displaystyle}
    
    \def\Re{{\rm Re \,}}
    \def\Im{{\rm Im \,}}
    \def\Ai{{\rm Ai \,}}

    \def\P2n{{\rm P}_{{\rm II}}^{(n)}}
    \def\R{\mathbb{R}}

    \newtheorem{theorem}{Theorem}[section]
    \newtheorem{lemma}[theorem]{Lemma}
    \newtheorem{corollary}[theorem]{Corollary}
    \newtheorem{proposition}[theorem]{Proposition}
    
    \newtheorem{Definition}[theorem]{Definition}
    
    \newtheorem{Remark}[theorem]{Remark}
    \newenvironment{remark}{\begin{Remark}\rm}{\end{Remark}}
    \newtheorem{Example}[theorem]{Example}
    
    \newtheorem{Assumptions}[theorem]{Assumptions}

    \DeclareMathOperator*{\Tr}{Tr}

\usepackage{tikz}
\usetikzlibrary{arrows}
\usetikzlibrary{decorations.pathmorphing}
\usetikzlibrary{decorations.markings}
\usetikzlibrary{patterns}
\usetikzlibrary{automata}
\usetikzlibrary{positioning}
\usepackage{tikz-cd}

\usepackage{pgfplots}

 \usepackage{authblk}

\date{}                     

\begin{document}
\title{A Riemann-Hilbert approach to the lower tail of the KPZ equation}
\author[1]{Mattia Cafasso}
\author[2]{Tom Claeys}
\affil[1]{\textit{LAREMA, UMR 6093, UNIV Angers, CNRS, SFR Math-STIC, France;} \texttt{cafasso@math.univ-angers.fr}}
\affil[2]{\textit{Institut de Recherche en Math\'ematique et Physique,  UCLouvain, Chemin du Cyclotron 2, B-1348 Louvain-La-Neuve, Belgium;} \texttt{tom.claeys@uclouvain.be}}

\maketitle

\begin{abstract}
Fredholm determinants associated to deformations of the Airy kernel are closely connected to the solution to the Kardar-Parisi-Zhang (KPZ) equation with narrow wedge initial data, and they also appear as largest particle distribution in models of positive-temperature free fermions.

We show that logarithmic derivatives of the Fredholm determinants can be expressed in terms of a $2\times 2$ Riemann-Hilbert problem, and we use this to derive asymptotics for the Fredholm determinants. As an application of our result, we derive precise lower tail asymptotics for the solution of the KPZ equation with narrow wedge initial data, refining recent results by Corwin and Ghosal \cite{CorwinGhosal}.\\

\end{abstract}

\tableofcontents

\newpage

\section{Introduction}

The KPZ equation is a stochastic PDE, introduced by Kardar, Parisi, and Zhang in 1986 \cite{KPZ}, which serves as a universal model for interface growth and which has a variety of physical applications, including liquid cristals, coffee stains, burning fronts, and bacterial colony growth, see e.g.\ \cite{bacterialcol1, burning, crystal, bacterialcol2, coffee} and references in the review article \cite{HHT}.
The equation is given by
\begin{equation}\label{eq:KPZ}
\partial_T\mathcal H(T,X)=\frac{1}{2}\partial_X^2\mathcal H(T,X)+\frac{1}{2}(\partial_X\mathcal H(T,X))^2+\xi(T,X),
\end{equation}
where $\xi(T,X)$ is space-time white noise. While it is not at all obvious what it means mathematically to be a solution of this equation, Bertini and Giacomin \cite{BertiniGiacomin} gave a precise meaning to this via the Cole-Hopf transformation $H(T,X)=\log Z(T,X)$, which relates the KPZ equation to the stochastic heat equation 
\[\partial_T Z(T,X)=\frac{1}{2}\partial_X^2Z(T,X)+Z(T,X)\xi(T,X).\] 
More recently, Hairer \cite{Hairer} extended this notion of a KPZ solution without passing through the Cole-Hopf transformation.
One of the physically relevant solutions of the KPZ equations, which can be defined using the Cole-Hopf transformation, is the one with {\em narrow wedge initial data}, formally given by
\[\mathcal H(0,X)=\log Z(0,X)\ \mbox{with}\ Z(0,X)=\delta_{X=0}.\]
We refer to \cite{Corwin, Corwin2} for more details about this solution and to \cite{Quastel} for a general discussion on Cole-Hopf solutions of KPZ.
It was proved by Amir, Corwin, and Quastel in 2010 \cite{AmirCorwinQuastel} that
the probability distribution of this solution can be characterized in terms of the Fredholm determinant of a deformed Airy kernel. 
Around the same time, Sasamoto and Spohn \cite{SasamotoSpohn} independently obtained the same result using similar methods but without rigorously verifying certain steps, and Dotsenko \cite{Dotsenko} and Calabrese, Le Doussal, and Rosso \cite{CLDR} predicted and confirmed this using non-rigorous methods.
A key element in the proofs of \cite{AmirCorwinQuastel, SasamotoSpohn} was the relation between the Cole-Hopf solutions of the KPZ equation and asymmetric exclusion processes, for which similar connections with Airy kernel determinants had been established \cite{Johansson1, TW-ASEP}.

\medskip

The Airy point process or Airy ensemble \cite{PraehoferSpohn, Soshnikov} is a determinantal point process on the real line which describes the largest particles in a wide class of random matrix ensembles and other interacting particle processes. Correlation functions in this process are expressed as determinants of the Airy kernel
\begin{equation}\label{Airykernel}
	K^\Ai(u, v) = \frac{ \Ai(u) \Ai'(v) - \Ai'(u) \Ai(v) }{ u - v },
\end{equation}
 where $\Ai$ denotes the Airy function.
The Airy point process has almost surely an infinite number of particles and a largest particle. Denoting $\zeta_1\geq\zeta_2\geq\cdots$ for the ordered random points, for any test function $\psi:\mathbb R\to\mathbb R$ which is such that 
the integral operator with kernel $\psi(x)K^\Ai(x,y)$ is trace-class, we have the identity
 \begin{equation}\label{DPPFredholm}
 \mathbb E_\Ai\left[\prod_{j=1}^\infty(1-\psi(\zeta_j))\right]=\det(1-\psi K^\Ai),
 \end{equation}
where $\det$ at the right hand side denotes the Fredholm determinant associated to the integral operator, and $\mathbb E_\Ai$ denotes the expectation with respect to the Airy point process.
In particular, the largest particle distribution is the Tracy-Widom distribution which can be expressed as the Fredholm determinant
$F^{\rm TW}(x):=\det\left(1-{\mathbf 1}_{(x,+\infty)}K^\Ai\right)$ of the Airy kernel operator acting on $L^2(x,+\infty)$.\\

In order to study the probability distribution of the KPZ solution with narrow wedge initial data $\mathcal H(T,X)$, it suffices to study $\mathcal H(T,0)$, since the probability distribution of $\mathcal H(T,X)-\frac{X^2}{2T}$ is independent of $X$ (see Theorem 1.1 in \cite{Corwin}).
For our purposes, it will be convenient to consider the re-scaled random variable
\begin{equation}\label{Upsilon}\Upsilon_T=\frac{\mathcal H(2T,0)+\frac{T}{12}}{T^{1/3}}.\end{equation}
Borodin and Gorin \cite{BorodinGorin} proved that the results from \cite{AmirCorwinQuastel} imply the identity
\begin{equation}\label{eq:BorodinGorin}
Q(s,T):=	\mathbb E_{\mathrm{KPZ}}\left[{\rm e}^{-{\rm e}^{T^{1/3}(\Upsilon_T+s)}}\right]=\mathbb E_\Ai\left[\prod_{j=1}^\infty\frac{1}{1+{\rm e}^{T^{1/3}(\zeta_j+s)}}\right],\end{equation} 
where $\zeta_j$, $j\in\mathbb N$, are the random points in the Airy point process, the expectation in the middle part is with respect to the KPZ equation, and  the expectation at the right is with respect to the Airy point process. 
$Q(s,T)$ can be interpreted as the Laplace transform of the associated solution to the stochastic heat equation.
As a direct consequence of \eqref{eq:BorodinGorin}, we observe that 
$Q(s,T)$ is the probability that there are no particles in the {\em thinned Airy point process} obtained by removing each particle $\zeta_k$, $k\in\mathbb N$, in the Airy point process independently with position-dependent probability $\left(1+{\rm e}^{T^{1/3}(\zeta_k+s)}\right)^{-1}$, i.e.\ the removal probability is close to $0$ for large positive values of $\zeta_j$, and close to $1$ for large negative values of $\zeta_j$.

Since the limiting density of the Airy point process is given by a square root law, in the sense that \cite{Soshnikov}
\begin{equation}
\label{eq:Airydensity}
\mathbb E_\Ai \Big[ \# \big\{j\in\mathbb N:\zeta_j>-r\big\} \Big] \sim\int_{-r}^{0} \frac{|\rho|^{1/2}}{\pi}d\rho=\frac{2}{3\pi}r^{3/2},\qquad\mbox{as $r\to\infty$,}
\end{equation}
one could intuitively expect that $\log Q(s,T)$ can be approximated for large $s$ by
\begin{multline*}\log Q(s,T) = \log \mathbb E_\Ai \left[ {\rm e}^{-\sum_{j=1}^\infty \log\left(1+{\rm e}^{T^{1/3}(\zeta_j+s)}\right)} \right]\\
\approx -\frac{1}{\pi}\int_{-\infty}^0\log\left(1+{\rm e}^{T^{1/3}(\zeta+s)}\right)
	|\zeta|^{1/2}d\zeta=-\frac{2T^{1/3}}{3\pi} \int_{-\infty}^0 \frac{1}{1 + {\rm e}^{-T^{1/3}(\zeta + s)}} |\zeta|^{3/2} d\zeta,	
	\end{multline*}
	where we used integration by parts in the last step. For large $T$, $\frac{1}{1 + {\rm e}^{-T^{1/3}(\zeta + s)}}$ approximates the step function ${\bf 1}_{(-s, +\infty)}(\zeta)$, and this suggests that $Q(s,T)$ can be approximated for large $s,T$ by
\begin{equation}\label{eq:naiveestimate}
\log Q(s,T)\approx -\frac{2T^{1/3}}{3\pi} \int_{-s}^{0}|\zeta|^{3/2} d\zeta =-\frac{4}{15\pi}T^{1/3}s^{5/2}.
\end{equation}
This heuristic picture appears to be accurate when $s$ is large compared to $T$ (see Remark \ref{rem:cases} below), but not in general when $s, T\to \infty$.\\

From \eqref{DPPFredholm}, it follows that $Q(s,T)$ is a Fredholm determinant:
\begin{equation}\label{eq:Fredholmid1}
	Q(s,T) =\det\Big(1-\sigma(T^{1/3}(x+s)) K^\Ai(x,y)\Big),\qquad \sigma(r) :=\frac{1}{1+{\rm e}^{-r}}.
\end{equation}
As $\sigma(T^{1/3}(x+s))$ approximates the step function ${\bf 1}_{(-s,\infty)}(x)$ as $T\to \infty$, this indicates that
the deformed Airy kernel Fredholm determinant in \eqref{eq:Fredholmid1} converges to the Tracy-Widom distribution $F^{\rm TW}(-s)$ as $T\to\infty$, which was indeed proved in \cite{AmirCorwinQuastel, Johansson}. 
There is a second Fredholm determinant representation for $Q(s,T)$, given by
\begin{equation}
\label{eq:Fredholmid2}
Q(s,T)=\det(1-K^\Ai_{T})_{L^2(-s,+\infty)},
\end{equation}
where $K^\Ai_T$ is the integral operator associated to the deformed Airy kernel
\begin{equation}\label{finiteTAirykernel2}
K^\Ai_{T}(u,v)=\int_{-\infty}^{+\infty}  \sigma(T^{1/3}r)\Ai(u+r)\Ai(v+r) dr.
\end{equation}
The equivalence between \eqref{eq:Fredholmid1} and \eqref{eq:Fredholmid2} is shown e.g.\ in \cite[Proof of Proposition 5.1]{AmirCorwinQuastel} for a general class of probability distributions $\sigma$.

The deformed Airy kernel \eqref{finiteTAirykernel2}
appeared also in a model of finite temperature free fermions, which is equivalent to the MNS matrix model. This model exhibits a critical scaling limit when the temperature tends to infinity together with the matrix dimension, and the limit distribution of the largest particle in the model then converges to $\det(1-K^\Ai_T)$ \cite{DLMS, Johansson, LiechtyWang}. {The same kernel appears also in the study of the edge scaling limit of the periodic Schur process \cite{BeteaBouttier}.}
Determinants and Pfaffians of a similar nature as $Q(s,T)$ appear in other models of statistical physics as well, like the stochastic six-vertex model, half-space KPZ, and asymmetric exclusion processes \cite{TW-ASEP, BBCW, BaikRains}.

\medskip

The asymptotic analysis of Fredholm determinants $\det(1-K)$ is rather standard in cases where the operator $K$ has small norm. In our setting, this is the case for $Q(s,T)$ as $s\to -\infty$, which encodes information about the upper tail of the KPZ solution.
As $s\to +\infty$, the norm of the deformed Airy kernel operator is not small, and we need more sophisticated methods to analyze $Q(s,T)$ and the associated lower tail of the KPZ equation asymptotically. 
A method developed by Its, Izergin, Korepin, and Slavnov \cite{IIKS} allows one to express logarithmic derivatives of Fredholm determinants $\det(1-K)$ of integral operators with kernels of {\em integrable form}
\begin{equation}\label{intkernel}
	K(x,y) := \frac{f^t(x) h(y)}{x-y},
\end{equation}
where $f, h$ are $k$-dimensional column vectors such that $f^t(x)h(x)=0$, in terms of a matrix Riemann-Hilbert (RH) problem of size $k\times k$.
If one manages to compute asymptotics for the solution of this RH problem, which can potentially be done using Deift-Zhou steepest descent methods \cite{DZ}, this provides a route towards asymptotics for the associated Fredholm determinants. 
This method has been implemented successfully for the determinant of the Airy kernel multiplied with piecewise constant functions \cite{DIK, BB, CharlierClaeys}, but the effectiveness of the method depends heavily on the situation at hand, especially when $k>2$.
It was observed in several works (see e.g. \cite[Section 2.1]{CorwinGhosal}, \cite[p2]{Tsai}, \cite[p2]{CGKLDT}) that the kernel \eqref{finiteTAirykernel2} is a limiting case of an integrable kernel with $k\to\infty$, which could lead to a representation of logarithmic derivatives in terms of an operator-valued RH problem instead of 
a matrix-valued one. Although there has been recent progress on such RH problems \cite{ItsKozlowski, ItsKozlowski2, Bothner}, extracting explicit asymptotics for \eqref{eq:Fredholmid2} as $s\to \infty$ from an operator-valued RH problem is currently beyond the state of the art of RH techniques.
Surprisingly, it apparently remained unnoticed until now that the kernel $\sigma(T^{1/3}(x+s))K^\Ai(x,y)$ from \eqref{eq:Fredholmid1} has a much more convenient form: it is also of integrable type, but now with $k=2$.
Indeed, if we set for instance
\begin{equation}\label{def:fgAiry}
	f(x) := \begin{pmatrix} -i\sigma(T^{1/3}(x+s))  \Ai'(x) \\ \sigma(T^{1/3}(x+s)) \Ai(x) \end{pmatrix} \quad h(x) := \begin{pmatrix} -i\Ai(x) \\ \Ai'(x) \end{pmatrix},
\end{equation}

then $\sigma(T^{1/3}(x+s))K^\Ai(x,y)$ is equal to \eqref{intkernel}.
Based on this observation, we will develop a $2\times 2$ RH approach for the determinant $Q(s,T)$. The RH problem related to this kernel does not have constant jumps, and for this reason we cannot relate it directly to an isomonodromic deformation of a linear system (Lax pair), like a Painlev\'e equation. This is not so suprising, since it is known \cite{AmirCorwinQuastel} that the deformed Airy kernel Fredholm determinants $Q(s,T)$ are connected to integro-differential generalizations of the Painlev\'e II equation, for which no Lax pair is known to the best of our knowledge.
One can in fact deduce this integro-differential Painlev\'e II equation from the RH problem, as we will show in a subsequent publication.
Perhaps more importantly, our RH problem is suitable for asymptotic analysis as $s\to \infty$, and allows us to obtain detailed asymptotics for $Q(s,T)$, which are uniform in the domain $M^{-1} < T < Ms^{3/2}$, for any $M > 0$ (and valid in particular for any fixed $T>0$). This will also yield uniform asymptotics for the lower tail of the probability distribution of the random variable $\Upsilon_T$, i.e. the rescaled KPZ solution with narrow wedge initial data.

\subsection*{Statement of results}

In order to describe the large $s$ asymptotics of $Q(s,T)$, we define
\begin{equation}\label{def:phi}
\phi(y) := \frac{4}{15\pi^6}(1+\pi^2y)^{5/2}-\frac{4}{15\pi^6}-\frac{2}{3\pi^4}y-\frac{1}{2\pi^2}y^2.
\end{equation}
After identifying $y$ with $-\zeta$, this is the large deviation rate function appearing in \cite{CGKLDT, SasorovMeersonProlhac, Tsai}, and it is straightforward to check that this function describes a cross-over between 
cubic behavior $\frac{y^3}{12}$ as $y\to 0$, and $5/2$ power-law behavior $\frac{4}{15\pi}y^{5/2}$ as $y\to\infty$.

\begin{theorem}\label{theorem:main}
Let $Q(s,T)$ be defined by \eqref{eq:BorodinGorin}. For any $M>0$, the asymptotics 
\begin{equation}\label{eq:Qlowertail}
\log Q(s,T)=-T^2\phi(sT^{-2/3}){-\frac{1}6}\sqrt{1+\frac{\pi^2s}{T^{2/3}}} 
+\mathcal O(\log^2 s)+\mathcal O(T^{1/3})
\end{equation}
hold uniformly in $M^{-1}\leq T \leq Ms^{3/2}$ as $s\to +\infty$.
\end{theorem}
The above results have straightforward implications for the lower tail of the solution to the KPZ equation with narrow wedge initial data.

\begin{corollary}
Let $\Upsilon_T$ be as defined in \eqref{Upsilon}.
There exist functions $A(s,T)$ and $B(s,T)$ such that 
\[A(s,T)\leq \log \mathbb P_{\mathrm{KPZ}}(\Upsilon_T<-s)\leq B(s,T),\]
and such that for any $\epsilon, M>0$, the asymptotics
\begin{equation}\label{eq:KPZlowertail1}
B(s,T)=-T^2\phi(sT^{-2/3}) {-\frac{1}6}\sqrt{1+\frac{\pi^2s}{T^{2/3}}}
+\mathcal O(\log^2 s)+\mathcal O(T^{1/3}) 
\end{equation}
and 
\begin{equation}\label{eq:KPZlowertail2}
A(s,T)=-T^2\phi(\tilde sT^{-2/3}){-\frac{1}6}\sqrt{1+\frac{\pi^2\tilde s}{T^{2/3}}}
+\mathcal O(\log^2 s)+\mathcal O(T^{1/3})
\end{equation}
hold uniformly in $M^{-1}\leq T\leq Ms^{3/2}$ as $s\to +\infty$, where we denoted
$\tilde s=s+(3+\epsilon)T^{-1/3}\log s$.
\end{corollary}
\begin{proof}
\begin{enumerate}
\item
Using Markov's inequality, we have
\begin{equation}\label{upperbound}
\mathbb P_{\mathrm{KPZ}}(\Upsilon_T\leq -s)=\mathbb P_{\mathrm{KPZ}}\left({\rm e}^{-{\rm e}^{T^{1/3}(\Upsilon_T+s)}}\geq{\rm e}^{-1}\right)\leq {\rm e}
\mathbb E_{\mathrm{KPZ}}\left[{\rm e}^{-{\rm e}^{T^{1/3}(\Upsilon_T+s)}}\right]
={\rm e} Q(s,T).
\end{equation} It follows that 
$$\log \mathbb P_{\mathrm{KPZ}}(\Upsilon_T\leq -s)\leq \log Q(s,T)+1,$$
and \eqref{eq:KPZlowertail1} now follows directly from Theorem \ref{theorem:main}.
\item
The lower bound \eqref{eq:KPZlowertail2} is obtained in a slightly different way, as in \cite[Section 3.1]{CorwinGhosal}: for any $s,\tilde s\in\mathbb R$, we have
\[Q(\tilde s,T)=\mathbb E_{\mathrm{KPZ}}\left[{\rm e}^{-{\rm e}^{T^{1/3}(\Upsilon_T+\tilde s)}}\right]\leq \mathbb E_{\mathrm{KPZ}}\left[{\bf 1}_{\{\Upsilon_T+s\leq 0\}}+{\bf 1}_{\{\Upsilon_T+s> 0\}}{\rm e}^{-{\rm e}^{T^{1/3}(\tilde s-s)}}\right]\]
and this implies
\begin{equation}\label{lowerbound}\mathbb P_{\mathrm{KPZ}}(\Upsilon_T\leq -s)\geq Q(\tilde s,T)-{\rm e}^{-{\rm e}^{T^{1/3}(\tilde s-s)}}.\end{equation}
Now we choose $\tilde s=s+(3+\epsilon)T^{-1/3}\log s$ for $\epsilon>0$ arbitrary small, such that the second term is $\mathcal O({\rm e}^{-s^{3+\epsilon}})$.
If we substitute the asymptotics obtained in Theorem \ref{theorem:main} for $Q(\tilde s,T)$ as $\tilde s\to +\infty$, we see that this term is bigger than ${\rm e}^{-cs^3}$ for some $c>0$ and for $s$ sufficiently large, so that this term dominates the 
term of order $\mathcal O({\rm e}^{-s^{3+\epsilon}})$. In other words, we have
\[\mathbb P_{\mathrm{KPZ}}(\Upsilon_T\leq -s)\geq 2Q(\tilde s,T),\]
for $s$ sufficiently large,
and the lower bound \eqref{eq:KPZlowertail2} follows after taking logarithms.
\end{enumerate}
\end{proof}

\begin{remark}\label{rem:cases}
Let us look in more detail at some particular cases.
\begin{enumerate}
\item For any fixed $T>0$, the above results imply
\begin{align*}&\log Q(s,T)=-\frac{4 T^{1/3}s^{5/2}}{15 \pi}+\frac{T^{2/3}s^2}{2\pi^2}
-\frac{2Ts^{3/2}}{3 \pi^3} +\frac{2T^{4/3}s}{3\pi^4}
{-\frac{\pi s^{1/2}}{6 T^{1/3}}} -\frac{1}{2\pi^5}T^{5/3}s^{1/2}
+\mathcal O(\log^2 s),\\
&\log \mathbb P_{\mathrm{KPZ}}(\Upsilon_T<-s)=-\frac{4 T^{1/3}s^{5/2}}{15 \pi}+\frac{T^{2/3}s^2}{2\pi^2}+\mathcal O(s^{3/2}\log s)
\end{align*}
as $s\to\infty$. The leading order term is consistent with the heuristic prediction \eqref{eq:naiveestimate}.
\item As $s,T\to\infty$ in such a way that $sT^{-2/3}\to\infty$, we have
\begin{align*}&\log Q(s,T)=-\frac{4}{15 \pi} T^{1/3}s^{5/2}+\frac{1}{2\pi^2}T^{2/3}s^2
-\frac{2}{3 \pi^3} Ts^{3/2}+\frac{2}{3\pi^4}T^{4/3}s\\&\qquad
-\frac{1}{2\pi^5}T^{5/3}s^{1/2}
+\frac{4}{15 \pi^6} T^2
{-\frac{\pi}{6 T^{1/3}}s^{1/2}} +\mathcal O(T^{7/3}s^{-1/2})+\mathcal O(\log^2 s)+\mathcal O(T^{1/3}),
\\& \log \mathbb P_{\mathrm{KPZ}}(\Upsilon_T<-s)=-\frac{4}{15 \pi} T^{1/3}s^{5/2}+\frac{1}{2\pi^2}T^{2/3}s^2+\mathcal O(s^{3/2}\log s).
\end{align*}
This improves upper and lower bounds obtained by Corwin and Ghosal \cite{CorwinGhosal}.
\item Setting $y=sT^{-2/3}$, we have for any fixed $y\in(0,+\infty)$ that
\[\log Q(y T^{2/3},T)=-T^2\phi(y)
+\mathcal O(s^{1/2}),\quad 
\log \mathbb P_{\mathrm{KPZ}}(\Upsilon_T<-s)=-T^2\phi(y)+\mathcal O(s^{3/2}\log s),
\]
which strengthens the recently obtained large deviation results by Tsai \cite{Tsai} (proving predictions from \cite{SasorovMeersonProlhac, CGKLDT}) by quantifying the error term.
\item Our results do not cover the regime where $s,T\to\infty$ in such a way that $sT^{-2/3}\to 0$. However, we can formally let $sT^{-2/3}\to 0$ at the right hand side of \eqref{eq:Qlowertail}, which yields
$-\frac{s^3}{12}+o(s^3),$
which is coherent with the lower tail expansion of the Tracy-Widom distribution given by \cite{DIK, BBdF}
\[{\log}F^{\rm TW}(-s)={\log}Q(s,\infty)=-\frac{s^3}{12}-\frac{1}{8}\log s+\frac{1}{24}\log 2 +\zeta'(-1)+o(1), \qquad \mbox{as $s\to\infty$,}\]
(where $\zeta$ is the Riemann $\zeta$-function)
and with the upper and lower bounds from \cite{CorwinGhosal}.
\end{enumerate}
\end{remark}

The rest of the paper is organized as follows. 
In the second section we express $Q(s,T)$ in terms of a $2\times 2$ RH problem, resulting in Theorem \ref{diffidentity}. Section 3 is devoted to the construction of the $g$-function, an essential ingredient for the non--linear steepest descent analysis of the RH problem, which is realized in Section 4. Finally, in Section 5, we use the results of the second and fourth section to prove Theorem \ref{theorem:main}.

\section{RH characterization of $Q(s,T)$}

In this section, we will express the logarithmic partial derivatives $\partial_T\log Q(s,T)$ and $\partial_s\log Q(s,T)$, with $Q(s,T)$ defined in equation \eqref{eq:Fredholmid1}, in terms of the solution of the following $2\times 2$ RH problem, which depends on parameters $s, \zeta_0\in\mathbb R$ and $T > 0$.
\begin{figure}[t]
\begin{center}
    \setlength{\unitlength}{0.6truemm}
    \begin{picture}(100,98.5)(0,2.5)

    \put(50,50){\thicklines\circle*{2}}
    \put(51,52){\small $\zeta_0$}
    \put(50,0){\line(0,1){100}}
    \put(0,50){\line(1,0){100}}
    \put(27,50){\thicklines\vector(1,0){.0001}}
    \put(77,50){\thicklines\vector(1,0){.0001}}
    \put(50,25){\thicklines\vector(0,1){.0001}}
      \put(50,75){\thicklines\vector(0,-1){.0001}}

    \put(52,10){\small $\begin{pmatrix} 1 & 0 \\ (1-\sigma(T^{1/3}(\zeta+s)))^{-1} & 1 \end{pmatrix}$}
    \put(52,85){\small $\begin{pmatrix} 1 & 0 \\ (1-\sigma(T^{1/3}(\zeta+s)))^{-1} & 1 \end{pmatrix}$}
\put(-70,57){\tiny $\begin{pmatrix} 0 & 1-\sigma(T^{1/3}(\zeta+s)) \\ -(1-\sigma(T^{1/3}(\zeta+s)))^{-1} & 0 \end{pmatrix}$}
    \put(100,48){\small $\begin{pmatrix} 1 & 1-\sigma(T^{1/3}(\zeta+s)) \\ 0 & 1 \end{pmatrix}$}    
    \end{picture}
    \caption{The jump contour $\Gamma$ and the jump matrices for $\Psi$.}
    \label{fig:contourPsi}
\end{center}
\end{figure}
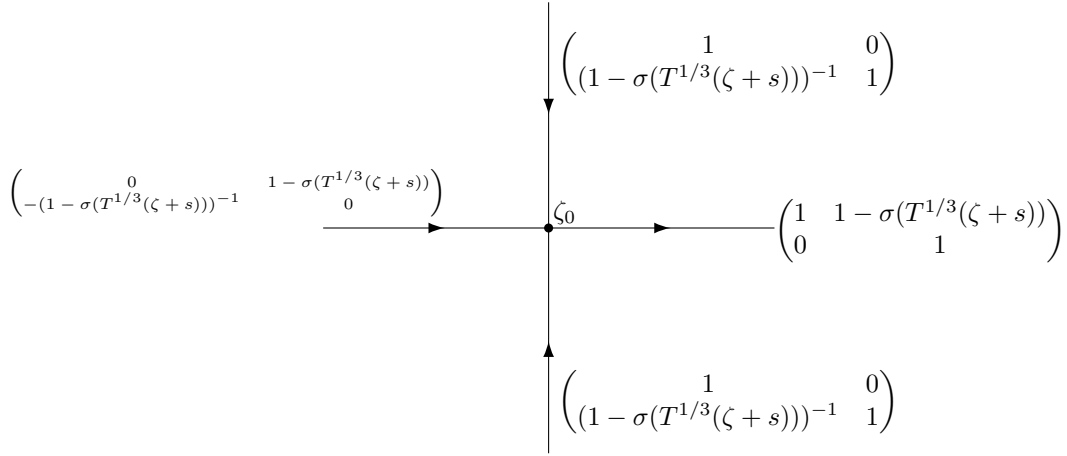
\subsubsection*{RH problem for $\Psi$}\label{RHPsi}
\begin{itemize}
\item[(a)] $\Psi : \mathbb C \backslash \Gamma \rightarrow \mathbb C^{2\times 2}$ is analytic, with
\begin{equation}\label{eq:defGamma}
\Gamma=\zeta_0 + (i\R \cup \R), 
\end{equation}
where the real line is oriented from left to right, and the two rays $\zeta_0+i\mathbb \R^\pm$ are oriented pointing to $\zeta_0$, see Figure \ref{fig:contourPsi}.
\item[(b)] $\Psi(\zeta)$ has continuous boundary values as $\zeta\in\Gamma\backslash \{\zeta_0\}$ is approached from the left ($+$ side) or from the right ($-$ side) and they are related by
\begin{equation}\label{jumpPsi}
\begin{array}{ll}
\Psi_+(\zeta) = \Psi_-(\zeta) \begin{pmatrix} 1 & 0 \\ (1-\sigma(T^{1/3}(\zeta+s)))^{-1} & 1 \end{pmatrix} & \textrm{for } \zeta \in   \zeta_0 + i\R^\pm,\\
\Psi_+(\zeta) = \Psi_-(\zeta) \begin{pmatrix} 0 & 1-\sigma(T^{1/3}(\zeta+s)) \\ -(1-\sigma(T^{1/3}(\zeta+s)))^{-1} & 0 \end{pmatrix} & \textrm{for } \zeta <\zeta_0, \\
\Psi_+(\zeta) = \Psi_-(\zeta) \begin{pmatrix} 1 & 1-\sigma(T^{1/3}(\zeta+s)) \\ 0 & 1 \end{pmatrix} & \textrm{for } \zeta>\zeta_0.
\end{array}
\end{equation}
\item[(c)] As $\zeta \rightarrow \infty$, $\Psi$ has the asymptotic behavior
\begin{equation}
\label{eq:psiasympinf}
\Psi(\zeta) = \left( I + \mathcal O\left(\frac{1}{\zeta}\right) \right) \zeta^{\frac{1}{4} \sigma_3} A^{-1} {\rm e}^{-\frac{2}{3}\zeta^{3/2}\sigma_3},
\end{equation}
where $A = (I + i \sigma_1) / \sqrt{2}$, $\sigma_1 = \begin{pmatrix} 0 & 1 \\ 1 & 0 \end{pmatrix}$, $\sigma_3 = \begin{pmatrix} 1 & 0 \\ 0 & -1 \end{pmatrix}$,
and where the principal branches of $\zeta^{\sigma_3/4}$ and $\zeta^{3/2}$ are taken, analytic in $\mathbb C\setminus(-\infty,0]$ and positive for $\zeta>0$.
\item[(d)] As $\zeta\to\zeta_0$, $\Psi(\zeta)$ remains bounded.
\end{itemize}

The main result in this section relates the logarithmic derivatives of $Q(s,T)$ with the solution $\Psi(\zeta)=\Psi(\zeta;s,T,\zeta_0)$ of the above RH problem. 
 \begin{theorem}\label{diffidentity}
The RH problem for $\Psi=\Psi(\zeta;s,T,\zeta_0)$ is uniquely solvable for any $s,\zeta_0\in\mathbb R$ and $T> 0$. Moreover, the following formulas hold true for any choice of $\zeta_0\in\mathbb R$:  
\begin{equation}\label{diffidentityeq}
	\partial_T\log Q(s,T)= -\frac{1}{6\pi iT^{2/3}} \int_{\mathbb R} (\zeta+s) \sigma'(T^{1/3}(\zeta+s))\left(\widehat\Psi^{-1}\widehat\Psi'\right)_{2,1}(\zeta;s,T,\zeta_0)d\zeta,
\end{equation} 
\begin{equation}\label{diffidentityeq_s}
		\partial_s\log Q(s,T)= -\frac{T^{1/3}}{2\pi i} \int_{\mathbb R} \sigma'(T^{1/3}(\zeta+s))\left(\widehat\Psi^{-1}\widehat\Psi'\right)_{2,1}(\zeta;s,T,\zeta_0)d\zeta,
	\end{equation} 
where $\widehat\Psi(\zeta)$ is given by
\begin{equation}\label{def:hatPsi}
\widehat\Psi(\zeta) :=\begin{cases}\Psi_+(\zeta),&\mbox{for $\zeta>\zeta_0$,}\\
\Psi_+(\zeta)\begin{pmatrix} 1 & 0 \\ (1-\sigma(T^{1/3}(\zeta+s)))^{-1} & 1 \end{pmatrix},&\mbox{for $\zeta<\zeta_0$,}\end{cases}
\end{equation}
and where $\widehat\Psi'$ is the $\zeta$-derivative of $\widehat\Psi$.
 \end{theorem}

The remaining part of this section is devoted to the proof of this result.\\

The uniqueness part of the theorem follows from standard arguments for RH problems: suppose that there exist two solutions $\Psi, \tilde\Psi$ to the above RH problem, then one shows first that $\det\Psi$ is an entire function using the fact that all the jump matrices have determinant $1$ and using condition (d). By condition (c), $\det\Psi(\zeta)$ tends to $1$ as $\zeta\to\infty$, and then Liouville's theorem implies that $\det\Psi\equiv 1$. Hence $\Psi$ is invertible and we can define the matrix-valued function 
$U(\zeta)=\tilde\Psi(\zeta)\Psi^{-1}(\zeta)$. It is straightforward to verify that condition (b) of the RH problem implies that $U$ has no jumps on $\Gamma$ and is hence, also by condition (d), an entire function. Moreover, by condition (c), $U(\zeta)$ tends to $I$ as $\zeta\to\infty$.
Again using Liouville's theorem, we obtain that $U\equiv I$ and $\Psi=\tilde\Psi$, which proves the uniqueness of the RH solution.\\

We will now show that the unique RH solution $\Psi$ exists, and after that we will relate it to $Q(s,T)$.
Recall that $Q(s,T)$ is the Fredholm determinant of an integrable operator, in the sense of Its, Izergin, Korepin, and Slavnov \cite{IIKS}, i.e. the determinant of the identity minus an operator of the form \eqref{intkernel} with $f$ and $h$ given in equation \eqref{def:fgAiry}. Following the general theory described in \cite{IIKS} and also in \cite{DIZ}, we can construct the (unique) solution to the RH problem for $Y$ given below using the resolvent of the operator
$K$. We do not provide details here, but refer to \cite[Section 3]{ClaeysDoeraene} where the procedure was carried out in detail for an Airy kernel with discontinuities similar to the deformed Airy kernel we consider here; adapting it to our situation is straightforward.

\subsubsection*{RH problem for $Y$}\label{RHY}

\begin{itemize}
\item[(a)] $Y : \mathbb C \backslash \R \rightarrow \mathbb C^{2\times 2}$ is analytic.

\item[(b)] $Y(\zeta)$ has continuous boundary values as $\zeta\in\R$ is approached from above or below, and they are related by 
\[Y_+(\zeta)=Y_-(\zeta)J(\zeta),\qquad J(\zeta)=I-2\pi i f(\zeta)h^t(\zeta),\]
or more explicitly by
\begin{equation*}
Y_+(\zeta) = Y_-(\zeta) \begin{pmatrix} 1 - 
2 \pi i  \sigma(T^{1/3}(\zeta + s)) \Ai(\zeta)\Ai'(\zeta) & 2\pi i \sigma(T^{1/3}(\zeta + s))\Ai^2(\zeta)  \\ -2 \pi i \sigma(T^{1/3}(\zeta + s))\Ai'^2(\zeta) & 1 + 2 \pi i \sigma(T^{1/3}(\zeta + s))\Ai(\zeta)\Ai'(\zeta)  \end{pmatrix}. 
\end{equation*}
\item[(c)] As $\zeta \rightarrow \infty$, we have
\begin{equation}\label{asY}
	Y(\zeta) = I + \mathcal O(\zeta^{-1}).
\end{equation}
\end{itemize}

To make the connection between $Q(s,T)$ and the RH solution $Y$ precise, we use \cite[Theorem 2.1]{BC}.
Specializing this result to our particular situation and adapting it to our notations, it states (noting that the contour $\mathcal C$ from \cite[Theorem 2.1]{BC} relevant for us is the real line) that
\begin{multline}\label{BCformula}
	\partial_T \log Q(s,T) = \int_\R \Tr \Big(Y_-^{-1}(\zeta)Y_-'(\zeta) \partial_T J^{-1}(\zeta)\Big) \frac{d \zeta}{2 \pi i}   
									\\- \int_\R  \partial_T \Big( (f^t)'(\zeta) h(\zeta)\Big)  d\zeta - 2 \pi i \int_\R h^t(\zeta) f'(\zeta) \partial_T h^t(\zeta)f(\zeta) d \zeta,
\end{multline}
and the same result holds if we replace the $T$-derivatives by $s$-derivatives.
In the remaining part of this section, we focus on the $T$-derivative and the proof of \eqref{diffidentityeq}, but it is directly verified that all identities below (up to \eqref{eq:diffidlast}) remain valid when replacing $T$-derivatives by $s$-derivatives, and that \eqref{diffidentityeq_s} follows from the counterpart of \eqref{eq:diffidlast}.

\medskip

In the next step, we transform the RH problem for $Y$ in a way which is analogous to \cite[Section 3]{ClaeysDoeraene}. For this, we need the solution to the Airy model RH problem. Define
\begin{align}
	\Phi_{\rm I}(\zeta) &= -\sqrt{2 \pi }\begin{pmatrix} \Ai'(\zeta) &-\omega\Ai'(\omega^2 \zeta) \\ i\Ai(\zeta) & -i\omega^2\Ai(\omega^2 \zeta)  \end{pmatrix} , \label{eqA}\\
	\Phi_{\rm IV}(\zeta) &=  -\sqrt{2 \pi }\begin{pmatrix} \Ai'(\zeta) &\omega^2\Ai'(\omega \zeta) \\ i\Ai(\zeta) & i\omega\Ai(\omega \zeta)  \end{pmatrix}, \label{eqB}\\
	\Phi_{\rm II}(\zeta) &= \Phi_{\rm I}(\zeta) \begin{pmatrix} 1 & 0 \\ -1 & 1 \end{pmatrix} \label{eqC},\\
	\Phi_{\rm III}(\zeta) &= \Phi_{\rm IV}(\zeta) \begin{pmatrix} 1 & 0 \\ 1 & 1 \end{pmatrix} \label{eqD},
\end{align} 
where $\omega := {\rm e}^{\frac{2 \pi i}3}$. Next, we fix $\zeta_0\in\mathbb R$ and we define $\Phi(\zeta)$ to be equal to $\Phi_{\rm I}(\zeta)$ for $\zeta-\zeta_0$ in the first quadrant, to $\Phi_{\rm II}(\zeta)$ for $\zeta-\zeta_0$ in the second quadrant, to $\Phi_{\rm III}(\zeta)$ for $\zeta-\zeta_0$ in the third quadrant, and to $\Phi_{\rm IV}(\zeta)$ for $\zeta-\zeta_0$ in the fourth quadrant.
It can be verified using the properties of the Airy function (similarly to e.g.\ \cite[Section 7]{DKMVZ}) that $\Phi$ solves the following RH problem.

\subsubsection*{RH problem for $\Phi$}\label{RHPHI}

\begin{itemize}
\item[(a)] $\Phi : \mathbb C \backslash \Gamma \rightarrow \mathbb C^{2\times 2}$ is analytic, with $\Gamma$ as in \eqref{eq:defGamma}.

\item[(b)] $\Phi(\zeta)$ has continuous boundary values as $\zeta\in \Gamma\setminus\{\zeta_0\}$ is approached from the left or from the right, and they are related by
\begin{align}
	\Phi_+(\zeta) &= \Phi_-(\zeta) \begin{pmatrix} 1 & 0 \\ 1 & 1 \end{pmatrix}, & \zeta \in \zeta_0+i\R^{\pm}, \label{eq:jump1}\\
	\Phi_+(\zeta) &=\Phi_-(\zeta) \begin{pmatrix} 0 & 1 \\ -1 & 0 \end{pmatrix}, & \zeta <\zeta_0, \label{eq:jump2}\\
	\Phi_+(\zeta) &= \Phi_-(\zeta) \begin{pmatrix} 1 & 1 \\ 0 & 1 \end{pmatrix}, & \zeta >\zeta_0, \label{eq:jump3}
\end{align}
\item[(c)] As $\zeta \rightarrow \infty$, we have
\begin{equation}\label{asPhi}
	\Phi(\zeta) = \Big(1 + \mathcal O(\zeta^{-3/2}) \Big)\zeta^{\sigma_3/4} A^{-1} {\rm e}^{-\frac{2}3 \zeta^{3/2}\sigma_3},
\end{equation}
where the principal branches of $\zeta^{\sigma_3/4}$ and $\zeta^{3/2}$ are taken. Note that this is precisely the same as \eqref{eq:psiasympinf}.
\item[(d)] As $\zeta\to\zeta_0$, $\Phi(\zeta)$ remains bounded.
\end{itemize}

\begin{remark}Our choice of jump contour for the Airy model problem is somewhat unusual, for reasons that will become clear later on in Section 3; it is more common to take half-lines with arguments $\pm\frac{2\pi}{3}$ instead of $\zeta_0+i\mathbb R^\pm$. 
However, since equations \eqref{eq:jump1}--\eqref{eq:jump3} are proven using the algebraic properties of the Airy function, and since the asymptotics of the Airy function are valid on sectors of the form $|\arg(\zeta)| < \pi - \delta$ with $\delta$ arbitrary small, the RH conditions are not affected by this choice of jump contours.
\end{remark}

Now we define $\Psi(\zeta)=\Psi(\zeta;s,T,\zeta_0)$ by
\begin{multline}\label{def:Psi}
	\Psi(\zeta) =  Y(\zeta)\Phi(\zeta)\\ \times\ \begin{cases}
\begin{pmatrix}1&0\\1-(1-\sigma(T^{1/3}(\zeta + s)))^{-1}&1\end{pmatrix}&\mbox{for $\Re(\zeta-\zeta_0)<0$, $\Im(\zeta-\zeta_0)>0$,}\\
\begin{pmatrix}1&0\\-1+(1-\sigma(T^{1/3}(\zeta + s)))^{-1}&1\end{pmatrix}&\mbox{for $\Re(\zeta-\zeta_0)<0$, $\Im(\zeta-\zeta_0)<0$,}\\
I&\mbox{for $\Re(\zeta-\zeta_0)>0$}.
\end{cases}
\end{multline}
\begin{lemma}
	The matrix--valued function $\Psi$ defined in \eqref{def:Psi} satisfies the RH problem for $\Psi$.
\end{lemma}
\proof
\begin{itemize}
\item[(a)] To see that $\Psi$ is analytic in $\mathbb C\setminus\Gamma$, it suffices to note that $Y$ and $\Phi$ are analytic in this domain, and that $(1-\sigma(T^{1/3}(\zeta + s)))^{-1}$ is an entire function.
\item[(b)] For the jump relations, we start by verifying that 
\begin{equation}\label{deftildePsi}
\tilde\Psi(\zeta) := Y(\zeta) \Phi(\zeta)
\end{equation}
satisfies the following jump conditions: 
 \begin{align}
		\tilde\Psi_+(\zeta) &= \tilde\Psi_-(\zeta) \begin{pmatrix} 1 & 1 - \sigma(T^{1/3}(\zeta + s)) \\ 0 & 1 \end{pmatrix}, & \zeta > \zeta_0,  \label{J1}\\
		\tilde\Psi_+(\zeta) &= \tilde\Psi_-(\zeta) \begin{pmatrix} 1 & 0 \\ 1 & 1 \end{pmatrix},& \zeta \in \zeta_0 +i\R^{\pm}, \label{J2}\\
		\tilde\Psi_+(\zeta) &=\tilde\Psi_-(\zeta) \begin{pmatrix} \sigma(T^{1/3}(\zeta + s)) & 1 - \sigma(T^{1/3}(\zeta + s)) \\ -1 - \sigma(T^{1/3}(\zeta + s)) & \sigma(T^{1/3}(\zeta + s)) \end{pmatrix},&  \zeta < \zeta_0 . \label{J3}
	\end{align} 
Indeed, the first equation \eqref{J1} follows from the relations
\begin{align*}
 	\tilde\Psi_-^{-1}(\zeta) \tilde\Psi_+(\zeta) &= \Phi_{\rm IV}^{-1}(\zeta)Y_-^{-1}(\zeta)Y_+(\zeta) \Phi_{\rm I}(\zeta)\\
	 								&= \Phi_{\rm IV}^{-1}(\zeta) \Big( I - 2 \pi i f(\zeta)h^t(\zeta)\Big) \Phi_{\rm I}(\zeta)
 \end{align*}
and
 \begin{equation}\label{fh}
 	f(\zeta) = \frac{i}{\sqrt{2\pi}}\sigma(T^{1/3}(\zeta + s)) \Phi_{\rm IV}(\zeta) \begin{pmatrix} 1\\ 0 \end{pmatrix}, \quad h^t(\zeta) =-\frac{1}{\sqrt{2\pi}} \begin{pmatrix} 0 & 1 \end{pmatrix}  \Phi_{\rm I}^{-1}(\zeta),
 \end{equation}
 together with the jump relation \eqref{eq:jump3} for $\Phi$.
The second equation \eqref{J2} follows from the jump relations \eqref{eq:jump1} for $\Phi$, since $Y$ is analytic on this part of the contour. Finally, for \eqref{J3}, we have
 \[
 	\tilde\Psi_-^{-1}(\zeta) \tilde\Psi_+(\zeta) = \Phi_{\rm III}^{-1}(\zeta)Y_-^{-1}(\zeta )Y_+(\zeta) \Phi_{\rm II}(\zeta),\]
 	and the result follows after a straightforward computation by \eqref{eqC}, \eqref{eqD}, \eqref{fh}, and the jump relation for $Y$.
  Now, the jump relations for $\tilde\Psi$ together with 
 \begin{multline}
\label{def:Psi2}
\Psi(\zeta)=\tilde\Psi(\zeta)\\\times\ \begin{cases}
\begin{pmatrix}1&0\\1-(1-\sigma(T^{1/3}(\zeta + s)))^{-1}&1\end{pmatrix}&\mbox{for $\Re(\zeta-\zeta_0)<0$, $\Im(\zeta-\zeta_0)>0$,}\\
\begin{pmatrix}1&0\\-1+(1-\sigma(T^{1/3}(\zeta + s)))^{-1}&1\end{pmatrix}&\mbox{for $\Re(\zeta-\zeta_0)<0$, $\Im(\zeta-\zeta_0)<0$,}\\
I&\mbox{for $\Re(\zeta-\zeta_0)>0$}.
\end{cases}
\end{multline}
imply the required jump relations for $\Psi$.
\item[(c)]
The asymptotic condition (c) follows by substituting the asymptotics \eqref{asY} and \eqref{asPhi} for $Y(\zeta)$ and $\Phi(\zeta)$ as $\zeta\to\infty$ in \eqref{deftildePsi} and then in \eqref{def:Psi2}. Additionally, we have that
$$1-(1-\sigma(T^{1/3}(\zeta + s)))^{-1} = -\exp(T^{1/3}(\zeta + s)),$$
remains bounded as $\zeta\to\infty$ with $\zeta-\zeta_0$ in the second or third quadrant, for any choice of $\zeta_0\in\mathbb R$, such that
$$
{\rm e}^{-\frac{2}{3} \zeta^{3/2}\sigma_3} \begin{pmatrix}1&0\\ \pm (1-(1-\sigma(T^{1/3}(\zeta + s)))^{-1})&1\end{pmatrix} = (I + \mathcal O(\zeta^{-1})){\rm e}^{-\frac{2}{3} \zeta^{3/2}\sigma_3}. 
$$
\item[(d)] As $\zeta\to \zeta_0$, $\Psi$ remains bounded because $Y$ and $\Phi$ do.
This completes the proof.
\end{itemize}
\qed\\
 
 The last step in the proof of Theorem \ref{diffidentity} consists of re-writing the formula \eqref{BCformula} in terms of of $\hat\Psi$ as defined in \eqref{def:hatPsi}. To this end, observe that, by $Y_+(\zeta) = Y_-(\zeta)J(\zeta)$ and by the cyclic property of the trace, we can write the first term at the right hand side of \eqref{BCformula} as 
	 	\begin{multline}
			\int_\R \Tr \Big(Y_-^{-1}(\zeta)Y_-'(\zeta)\partial_T J_T(\zeta) J^{-1}(\zeta)\Big) \frac{d \zeta}{2 \pi i} = \int_\R \Tr \Big(Y_+^{-1}(\zeta)Y_+'(\zeta) J^{-1}(\zeta)\partial_TJ(\zeta) \Big) \frac{d \zeta}{2 \pi i} \\
			- \int_\R \Tr \Big(J^{-1}(\zeta)J'(\zeta) J^{-1}(\zeta) \partial_TJ(\zeta) \Big) \frac{d\zeta}{2 \pi i }.
		\end{multline}
Moreover, using the fact that $J(\zeta) = I - 2 \pi i f(\zeta)h^t(\zeta)$ and again the cyclic property of the trace it is straightforward to verify that\footnote{This is in fact the content of Remark A.1 in \cite{BC}, except for a factor $2 \pi i $ that is missing there.}
$$
	\int_\R \Tr \Big(J^{-1}(\zeta)J'(\zeta) J^{-1}(\zeta)\partial_TJ(\zeta) \Big) \frac{d\zeta}{2 \pi i } = -4 \pi i \int_\R h^t(\zeta) f'(\zeta) \partial_Th(\zeta) f(\zeta) d \zeta,
$$
and the right hand side is zero since $h$ does not depend on $T$ (nor on $s$). Hence, in our particular case, we can re-write \eqref{BCformula} as
\begin{equation}\label{BCformula2}
	\partial_T\log Q(s,T)= \int_\R \Tr \Big(Y_+^{-1}(\zeta)Y_+'(\zeta) J^{-1}(\zeta) \partial_T J(\zeta) \Big) \frac{d \zeta}{2 \pi i} -  \int_\R  \partial_T \Big( (f^t)'(\zeta) h(\zeta) \Big)  d\zeta.
\end{equation}
We now compute explicitly the two terms at the right hand side. From the definition \eqref{def:fgAiry} of $f$ and $h$, we get
\begin{equation}\label{secondterm} 
	\int_\R \Tr\left( \partial_T \Big( (f^t)'(\zeta) h(\zeta) \Big) \right) d\zeta = \int_\R \partial_T\sigma(T^{1/3}(\zeta + s))(\Ai'^2(\zeta) - \zeta\Ai^2(\zeta)) d\zeta.
\end{equation}
For the remaining term, we have
\begin{multline*}
	J^{-1}(\zeta)\partial_T J(\zeta) = -2 \pi i \partial_T\sigma(T^{1/3}(\zeta + s)) \begin{pmatrix} \Ai(\zeta)\Ai'(\zeta) & -\Ai^2(\zeta) \\ \Ai'^2(\zeta) & -  \Ai(\zeta)\Ai'(\zeta) \end{pmatrix} \\= -2 \pi i \frac{\partial_T\sigma(T^{1/3}(\zeta + s))}{\sigma(T^{1/3}(\zeta + s))}f(\zeta)h^t(\zeta)
\end{multline*}
Now we use once more the cyclic property of the trace, \eqref{fh}, and \eqref{eqA}--\eqref{eqB}
to arrive at
\begin{align}
	&\int_\R \Tr \Bigg(Y_+^{-1}(\zeta)Y_+'(\zeta) J^{-1}(\zeta)\partial_T J(\zeta) \Bigg) \frac{d \zeta}{2 \pi i} \\&= 
	-\int_\R \partial_T\sigma(T^{1/3}(\zeta + s))\begin{pmatrix} 0 & 1 \end{pmatrix} \Phi_{\rm I}^{-1}(\zeta)Y_+^{-1}(\zeta)Y_+'(\zeta) \Phi_{\rm I}(\zeta)  \begin{pmatrix}1 \\ 0 \end{pmatrix} \frac{d \zeta}{2\pi i}.
\end{align}
Finally, we use the fact that $\widehat\Psi(\zeta) =  Y(\zeta)\Phi_{\rm I}(\zeta) $, which follows from \eqref{def:hatPsi}, \eqref{def:Psi}, and \eqref{eqC}, to obtain
\begin{multline}
	\int_\R \Tr \Bigg(Y_+^{-1}(\zeta)Y_+'(\zeta) J^{-1}(\zeta)\partial_T J(\zeta) \Bigg) \frac{d \zeta}{2 \pi i} \\ 
=	 -\frac{1}{2 \pi i} \int_\R \partial_T\sigma(T^{1/3}(\zeta + s)) \Big(\hat\Psi^{-1} \hat\Psi' \Big)_{21}(\zeta) d \zeta 
+ \int_\R \partial_T\sigma(T^{1/3}(\zeta + s))\Big(\Ai'^2(\zeta) - \zeta\Ai^2(\zeta) \Big) d\zeta, \label{eq:diffidlast}
\end{multline}
for which we also used \eqref{eqA}.
Combining this last equation with \eqref{BCformula2} and \eqref{secondterm}, we obtain \eqref{diffidentityeq}. 

\section{Construction of the $g$-function}

\subsection{Heuristics}
Our objective in this section and the next one is to derive asymptotics for the solution $\Psi = \Psi(\zeta;s,T,\zeta_0)$ of the RH problem for $\Psi$ from the previous section as $s\to\infty$, uniform in $\zeta$ and in $M^{-1}\leq T\leq Ms^{3/2}$, with  $M>0$ arbitrary large, for a convenient choice of $\zeta_0$. To this end, we will transform the RH problem for $\Psi$, in several steps, to a RH problem for $R$ with jumps which are close to the identity matrix as $s\to\infty$, and which is normalized at infinity such that $\lim_{\lambda\to\infty}R(\lambda)=I$. 
Such a RH problem is often referred to as a {\em small-norm} RH problem, and the general theory of RH problems will then allow us to conclude that $R(\lambda)$ is close to $I$ as $s\to\infty$, with estimates for the error term which will be uniform in $\lambda$ and in $M^{-1}\leq T\leq Ms^{3/2}$.\\

The main obstacles for us to obtain such a small-norm RH problem are:
\begin{itemize}
\item[(i)] the fact that the jump matrices for $\Psi$, given in \eqref{jumpPsi}, depend on $s$ and $T$ in a rather complicated manner,
\item[(ii)] the fact that $\Psi$ is not normalized to the identity matrix as $\zeta\to\infty$, see \eqref{eq:psiasympinf}.
\end{itemize}
In order to improve on the first issue (i), we will first apply a change of variables $\zeta+s=s\lambda$, which will remove the $s$-dependence in the jump matrices \eqref{jumpPsi}, at the expense of making the asymptotic behavior slightly more complicated. Indeed, the exponential factor ${\rm e}^{-\frac{2}{3}\zeta^{3/2}\sigma_3}$ in \eqref{eq:psiasympinf} for instance becomes
${\rm e}^{-\frac{2}{3}s^{3/2}\left(\lambda-1\right)^{3/2}\sigma_3}$, which behaves like 
${\rm e}^{-s^{3/2}\left(\frac{2}{3}\lambda^{3/2}-\lambda^{1/2}\right)\sigma_3}$ as $\lambda\to\infty$.
Next, in order to resolve issue (ii), we will need to regularize the asymptotic behavior of the RH solution as $\lambda\to\infty$, by introducing a conveniently defined $g$-function, and by multiplying the RH solution from the right by ${\rm e}^{s^{3/2}g(\lambda)\sigma_3}$. The construction of this $g$-function is the key to a successful RH analysis, and we devote this section to it.

Although some of the reasons behind our definition of the $g$-function will become clear only a posteriori, we attempt to give some heuristics behind it now.
First, we will need this function to behave for large $\lambda$ like $\frac{2}{3}\lambda^{3/2}-\lambda^{1/2}$, 
in such a way that ${\rm e}^{s^{3/2}g(\lambda)\sigma_3}$ can neutralize the factor ${\rm e}^{-\frac{2}{3}s^{3/2}\left(\lambda-1\right)^{3/2}\sigma_3}$ appearing in the large $\lambda$ asymptotics for the RH problem.
While this condition still leaves us with a lot of freedom to define $g$, we will secondly take the opportunity to construct $g$ in such a way that it helps for making the jump matrices as simple as possible. To do this, we will need ${\rm e}^{s^{3/2}g(\lambda)\sigma_3}$ to be analytic except on a half-line of the form $(-\infty,\lambda_0]$ and 
to interact in a suitable way with the function $1-\sigma(sT^{1/3}\lambda)$ appearing in the jump matrices of the RH problem.
More precisely, it will turn out convenient to have the property ${\rm e}^{s^{3/2}(g_+(\lambda)+g_-(\lambda))}=\frac{1-\sigma(sT^{1/3}\lambda)}{1-\sigma(sT^{1/3}\lambda_0)}$ for $\lambda<\lambda_0$.

Combining all the above conditions, we are led to the question of constructing a function $g(\lambda)=g(\lambda;s,T)$ satisfying the following properties.
\subsubsection*{RH problem for $g$}\label{RHg}
\begin{itemize}
\item[(a)] $g : \mathbb C\setminus (-\infty,\lambda_0] \longrightarrow \mathbb C$ is analytic.
\item[(b)] $g(\lambda)$  has continuous boundary values as $\lambda \in (-\infty,\lambda_0]$ is approached from above or below, and they are related by
\begin{equation}\label{jumpg}
g_+(\lambda)+g_-(\lambda)=V(\lambda;s,T)-V(\lambda_0;s,T),\qquad\mbox{for }\lambda \in (-\infty,\lambda_0),
\end{equation}
where we denote
\begin{equation}
\label{def:h}V(\lambda;s,T) := s^{-3/2}\log\left(1-\sigma(sT^{1/3}\lambda)\right).
\end{equation}
\item[(c)] For fixed $s,T>0$, $g$ has the asymptotics
\begin{equation}\label{as:g}g(\lambda)=\frac{2}{3}\lambda^{3/2}-\lambda^{1/2}-\frac{V(\lambda_0;s,T)}{2}+g_1\lambda^{-1/2}
+o(\lambda^{-1/2})\qquad\mbox{as $\lambda\to\infty$,}\end{equation}
for a certain value $g_1$, which may depend on $s, T$. Here $\lambda^{3/2}$ and $\lambda^{\pm 1/2}$ are the principal branches of the root functions, analytic in $\mathbb C\setminus(-\infty,0]$ and positive for $\lambda>0$.
\end{itemize}

We note that $V$ is negative, strictly decreasing in $\lambda$ and tending to $0$ as $\lambda\to -\infty$ and to $-\infty$ as $\lambda\to +\infty$. Moreover, by the definition \eqref{eq:Fredholmid1} of $\sigma$, we have 
\begin{equation}\label{eq:hprime}V'(\lambda;s,T)=-s^{-1/2}T^{1/3}\sigma(sT^{1/3}\lambda).\end{equation}
\subsection{Construction of $g'$}
Instead of constructing $g$ directly, we will first construct a function $g'$ (which will indeed turn out to be the derivative of $g$) solving the following RH problem. 
\subsubsection*{RH problem for $g'$}
\begin{itemize}
\item[(a)] $g'$ is analytic in $\mathbb C\setminus (-\infty,\lambda_0]$.
\item[(b)] $g'$ has the jump relation
\begin{equation}
g_+'(\lambda)+g_-'(\lambda)=V'(\lambda;s,T)\qquad\mbox{for }\lambda<\lambda_0.\label{eq:jumpgprime}
\end{equation}
\item[(c)] $g'$ has the asymptotics
\[g'(\lambda)=\lambda^{1/2}-\frac{1}{2}\lambda^{-1/2}-\frac{g_1}{2}\lambda^{-3/2}+o(\lambda^{-3/2})\qquad\mbox{as }\lambda\to\infty.\]
\end{itemize}

\begin{lemma}\label{lemma:g'}
	The function
	\begin{align}	g'(\lambda)&=(\lambda-\lambda_0)^{1/2}\left(1-\frac{1}{2\pi}\int_{-\infty}^{\lambda_0}\frac{V'(\xi;s,T)}{\sqrt{\lambda_0-\xi}}\frac{d\xi}{\xi-\lambda}\right) \nonumber\\	&=(\lambda-\lambda_0)^{1/2}\left(1+\frac{s^{-1/2}T^{1/3}}{2\pi}\int_{-\infty}^{\lambda_0}\frac{\sigma(sT^{1/3}\xi)}{\sqrt{\lambda_0-\xi}}\frac{d\xi}{\xi-\lambda}\right),\label{exprg}
	\end{align}
where $(.)^{1/2}$ denotes the principal branch of the complex square root, $\sqrt{.}$ is the positive square root, and $\lambda_0=\lambda_0(s,T)$ is the unique real solution to the equation
	\begin{equation}\label{eq:endpointequation}
		\lambda_0-1=-\frac{s^{-1/2}T^{1/3}}{\pi}\int_{-\infty}^{\lambda_0}\frac{\sigma(sT^{1/3}\xi)}{\sqrt{\lambda_0-\xi}}d\xi=-\frac{s^{-1/2}T^{1/3}}{\pi}\int_{-\infty}^{0}\frac{\sigma(sT^{1/3}(u+\lambda_0))}{\sqrt{-u}}du,
\end{equation}
solves the RH problem for $g'$.
\end{lemma}
\begin{proof}
It is immediate, using the properties of the Cauchy integral, to prove that $g'$, defined as in \eqref{exprg}, satisfies the jump condition \eqref{eq:jumpgprime}. Then, expanding the right hand side of the last expression in \eqref{exprg} we easily verify that \eqref{eq:endpointequation} needs to be satisfied if we want the coefficient of $\lambda^{-1/2}$ in the asymptotic expansion of $g'$ to be equal to -1/2.

Moreover, the left hand side of \eqref{eq:endpointequation} is linearly increasing in $\lambda_0$, while the right hand side is decreasing from $0$ to $-\infty$ as $\lambda_0$ goes from $-\infty$ to $+\infty$. It follows that there is a unique solution to \eqref{eq:endpointequation}.
\end{proof}
\begin{remark}\label{rem:unicitylambda0}
	It is easy to see that $\lambda_0(s,T)<1$ for any $s,T>0$, and that $\lambda_0(s,T)$ decreases as $T$ increases. 
\end{remark}
\begin{remark}
	An alternative representation for $g'$, obtained from \eqref{exprg} after integration by parts, is
\begin{equation}\label{exprg2}
g'(\lambda)=(\lambda-\lambda_0)^{1/2}+
\frac{s^{-1/2}T^{1/3}}{2\pi i}\int_{-\infty}^{\lambda_0}\log\left(\frac{(\lambda-\lambda_0)^{1/2}-i\sqrt{\lambda_0-\xi}}{(\lambda-\lambda_0)^{1/2}+i\sqrt{\lambda_0-\xi}}\right)d\sigma(sT^{1/3}\xi),
\end{equation}
where $\log$ is the principal branch of the complex logarithm.
\end{remark}
We now define 
\begin{equation}\label{def:g}
g(\lambda) := \int_{\lambda_0}^\lambda g'(\xi)d\xi,
\end{equation}
with $g'$ given by \eqref{exprg}, and
where the path on integration is always chosen within the domain of analyticity of $g'$.
It is straightforward to verify that $g$ defined by \eqref{def:g} is analytic in $\mathbb C\setminus(-\infty,\lambda_0]$,
satisfies the jump relation \eqref{jumpg}, and has asymptotics of the form
\[g(\lambda)=\frac{2}{3}\lambda^{3/2}-\lambda^{1/2}+C+g_1\lambda^{-1/2}
+o(\lambda^{-1/2})\qquad\mbox{as $\lambda\to\infty$,}\] 
for some integration constant $C$. This can only be consistent with the jump relation if $C=-\frac{V(\lambda_0;s,T)}{2}$, and
\eqref{as:g} follows. Hence we have solved the RH problem for $g$.

\subsection{Estimates for the endpoint $\lambda_0$}
We now establish the leading order asymptotics of $\lambda_0=\lambda_0(s,T)$.
\begin{proposition}\label{prop:asendpoint}
Let $\lambda_0=\lambda_0(s,T)$ be the unique real solution of the equation \eqref{eq:endpointequation}.
For any $M>0$ the following estimate holds uniformly in $M^{-1} \leq T\leq Ms^{3/2}$ as $s\to +\infty$:
\[\lambda_0=\frac{T^{2/3}}{\pi^2 s}\left(\sqrt{1+\pi^2sT^{-2/3}}-1\right)^2+\mathcal O(s^{-5/2}T^{-1/3}).\]
In particular, for any $M>0$, there exist $s_0, \kappa > 0$ such that $\lambda_0\in (\kappa, 1)$ for any $s\geq s_0$ and for any $M^{-1} \leq T\leq Ms^{3/2}$.
\end{proposition}
\begin{proof}
By \eqref{eq:endpointequation}, we have
\begin{align}
\left|\lambda_0-1+\frac{2T^{1/3}}{\pi s^{1/2}}\sqrt{\max\{0,\lambda_0\}}\right|&\leq \frac{s^{-1/2}T^{1/3}}{\pi}\int_{-\infty}^{\lambda_0}
\frac{\left|\sigma(sT^{1/3}\xi)-{\bf 1}_{(0,+\infty)}(\xi)\right|}{\sqrt{\lambda_0-\xi}}d\xi\nonumber\\
&=\frac{T^{1/6}}{\pi s}\int_{-\infty}^{sT^{1/3}\lambda_0}
\frac{\left|\sigma(u)-{\bf 1}_{(0,+\infty)}(u)\right|}{\sqrt{sT^{1/3}\lambda_0-u}}du.\label{eq:endpoint1}
\end{align}

We start by showing that $\lambda_0\geq 0$. Assuming by contra-position that $\lambda_0<0$, we obtain from the above inequality and
the inequality
\begin{equation}
\left|\sigma(x)-{\bf 1}_{(0,+\infty)}(x)\right|\leq {\rm e}^{-|x|},\qquad x\in\mathbb R \label{ineq:sigma}
\end{equation} that
\begin{multline*}
1\leq |\lambda_0-1|\leq \frac{T^{1/6}}{\pi s}\int_{-\infty}^{sT^{1/3}\lambda_0}\frac{{\rm e}^{u}du}{\sqrt{sT^{1/3}\lambda_0-u}}= \frac{T^{1/6}}{\pi s}\int_{-\infty}^{0}\frac{{\rm e}^{v+sT^{1/3}\lambda_0}}{\sqrt{-v}}dv\\ =\frac{T^{1/6}}{\sqrt{\pi} s}{\rm e}^{sT^{1/3}\lambda_0}\leq \frac{M^{1/6}}{\sqrt{\pi}s^{3/4}}.
\end{multline*}
This gives a contradiction for $s$ sufficiently large.

Next, as in \eqref{eq:endpoint1} but now with $\lambda_0>0$, we obtain by \eqref{eq:endpointequation} that
\[
\left|\lambda_0-1+\frac{2T^{1/3}}{\pi s^{1/2}}\sqrt{\lambda_0}\right|=
\frac{T^{1/6}}{\pi s}\left|
\int_{-\infty}^{sT^{1/3}\lambda_0}
\frac{\sigma(u)-{\bf 1}_{(0,+\infty)}(u)}{\sqrt{sT^{1/3}\lambda_0-u}}du \right|.\]
It is straightforward to show by \eqref{ineq:sigma} that the contribution of the region $|u|\geq 2\sqrt{s}\lambda_0$ to the latter integral is $\mathcal O({\rm e}^{-\sqrt{s}\lambda_0})$ as $s\to\infty$, uniformly in $T$.
Moreover, since  
$\int_{-2\sqrt{s}\lambda_0}^{2\sqrt{s}\lambda_0}
\frac{\sigma(u)-{\bf 1}_{(0,+\infty)}(u)}{\sqrt{|sT^{1/3}\lambda_0|}}du=0$ and since $\frac{T^{1/6}}{\pi s}\leq \frac{M^{1/6}}{\pi s^{3/4}}$ is small, we obtain that 
\begin{align*}
&\left|\lambda_0-1+\frac{2T^{1/3}}{\pi s^{1/2}}\sqrt{\lambda_0}\right|\nonumber\\ &\leq
\frac{T^{1/6}}{\pi s}\left|
\int_{-2\sqrt{s}\lambda_0}^{2\sqrt{s}\lambda_0}
\left(\sigma(u)-{\bf 1}_{(0,+\infty)}(u)\right)\left(\frac{1}{\sqrt{|sT^{1/3}\lambda_0-u|}}-\frac{1}{\sqrt{|sT^{1/3}\lambda_0|}}\right)du \right|
+\frac{1}{2}{\rm e}^{-\sqrt{s}\lambda_0}\nonumber\\
&= \frac{1}{\pi s^{3/2}}\left|
\int_{-2\sqrt{s}\lambda_0}^{2\sqrt{s}\lambda_0}
\left(\sigma(u)-{\bf 1}_{(0,+\infty)}(u)\right)\left(\left(1-\frac{u}{sT^{1/3}\lambda_0}\right)^{-1/2}-1\right)du \right|
+\frac{1}{2}{\rm e}^{-\sqrt{s}\lambda_0}
,\end{align*}
for $s$ sufficiently large.
Since $|(1-y)^{-1/2}-1|\leq |y|$ for $y$ sufficiently small, we have for $s$ sufficiently large that
\begin{equation}\left|\lambda_0-1+\frac{2T^{1/3}}{\pi s^{1/2}}\sqrt{\lambda_0}\right|\leq \frac{1}{\pi s^{5/2}T^{1/3}\lambda_0}\int_{-\infty}^{+\infty}\left|\left(\sigma(u)-{\bf 1}_{(0,+\infty)}(u)\right)u \right| du+\frac{1}{2}{\rm e}^{-\sqrt{s}\lambda_0}.\label{eq:boundlam0}\end{equation}
From this, we can first derive that $\lambda_0$ is bounded below by a positive constant independent of $s, T$. Indeed, if we suppose that $\lambda_0\to 0$, the left hand side would tend to $1$ and the second term at the right is smaller than $1/2$, which implies that the first term at the right cannot tend to $0$. In other words, $s^{5/2}T^{1/3}\lambda_0$ remains bounded, and $sT^{1/3}\lambda_0=\mathcal O(s^{-3/2})$. Substituting this in \eqref{eq:endpointequation}, we obtain
\[|\lambda_0-1|\leq 
\frac{s^{-1/2}T^{1/3}}{\pi}\int_{-\infty}^0\frac{\sigma(sT^{1/3}u+1)}{\sqrt{-u}}du\leq 
\frac{T^{1/6}}{\pi s}\int_{-\infty}^0\frac{\sigma(v+1)}{\sqrt{-v}}dv=\mathcal O(s^{-3/4}),
\]
which implies the contradiction that $\lambda_0$ would tend to $1$.

Knowing that $\lambda_0$ is bounded below by a positive constant, we obtain from \eqref{eq:boundlam0} that
\[\left|\lambda_0-1+\frac{2T^{1/3}}{\pi s^{1/2}}\sqrt{\lambda_0}\right|=\mathcal O(s^{-5/2}T^{-1/3})\]
as $s\to\infty$ uniformly in $T$,
and it is then straightforward to complete the proof (recall that the inequality $\lambda_0 < 1$ was already observed in Remark \ref{rem:unicitylambda0}).
\end{proof}

\subsection{Properties of $g$}

We start this subsection with a technical proposition which will be crucial in the RH analysis later on. 

\begin{proposition}\label{prop:ineqspsi}
The following inequalities are satisfied :
\begin{align}
&\label{ineq1}\left| {\rm e}^{-s^{3/2}\left(2g(\lambda)-V(\lambda)+V(\lambda_0)\right)}\right|\leq {\rm e}^{-\frac{4}{3}s^{3/2}(\lambda-\lambda_0)^{3/2}}&\mbox{for $\lambda>\lambda_0$,}\\
&\label{ineq2}\left| {\rm e}^{s^{3/2}\left(2g(\lambda)-V(\lambda)+V(\lambda_0)\right)}\right|\leq 2{\rm e}^{-\frac{2\sqrt{2}}{3}s^{3/2}|\lambda-\lambda_0|^{3/2}}
&\mbox{for $\lambda\in \lambda_0+i\mathbb R$.}
\end{align}
\end{proposition}
\begin{proof}
For the first assertion, we observe after a straightforward computation using \eqref{exprg} and the identity $$\sqrt{\lambda-\lambda_0}\int_{-\infty}^{\lambda_0}\frac{d\xi}{\sqrt{\lambda_0-\xi}(\xi-\lambda)}=-\pi$$
that
\begin{equation}\label{exprgprime}2g'(\lambda)-V'(\lambda)=2\sqrt{\lambda-\lambda_0}\left(1+\frac{s^{-1/2}T^{1/3}}{2\pi}\int_{-\infty}^{\lambda_0}\frac{\sigma(sT^{1/3}\xi)-\sigma(sT^{1/3}\lambda)}{\sqrt{\lambda_0-\xi}(\xi-\lambda)}d\xi\right)\geq 2\sqrt{\lambda-\lambda_0}\end{equation}
for $\lambda>\lambda_0$.
After integration, we obtain by \eqref{def:g} that
\[2g(\lambda)-V(\lambda)+V(\lambda_0)\geq 2\int_{\lambda_0}^{\lambda}\sqrt{\xi-\lambda_0}d\xi=\frac{4}{3}(\lambda-\lambda_0)^{3/2}\]
for $\lambda>\lambda_0$, and \eqref{ineq1} follows easily.

\medskip

For \eqref{ineq1}, the cases $\lambda\in i\mathbb R^+$ and $\lambda\in i\mathbb R^-$ are similar, and we restrict to $\lambda$ in the upper half plane.
By \eqref{def:h} and \eqref{eq:Fredholmid1}, we have
\begin{align*}
\left|{\rm e}^{s^{3/2}\left(2g(\lambda)-V(\lambda)+V(\lambda_0)\right)}\right|&=\left|{\rm e}^{2s^{3/2}g(\lambda)}\right|\frac{\left|1-\sigma(sT^{1/3}\lambda_0)\right|}{\left|1-\sigma(sT^{1/3}\lambda)\right|}=\left|{\rm e}^{2s^{3/2}g(\lambda)}\right|\frac{\left|1+{\rm e}^{sT^{1/3}\lambda}\right|}{\left|1+{\rm e}^{sT^{1/3}\lambda_0}\right|}\\
&\leq \left|{\rm e}^{2s^{3/2}g(\lambda)}\right|\frac{1+{\rm e}^{sT^{1/3}\lambda_0}}{{\rm e}^{sT^{1/3}\lambda_0}}\leq 2\left|{\rm e}^{2s^{3/2}g(\lambda)}\right|,
\end{align*}
where we used the fact that $\lambda_0>0$ in the last step.
It remains to show that 
$$\left|{\rm e}^{2s^{3/2}g(\lambda)}\right|\leq {\rm e}^{-\frac{2\sqrt{2}}{3}s^{3/2}|\lambda-\lambda_0|^{3/2}}.$$
To see this, we use \eqref{def:g}
and \eqref{exprg2} and obtain, using the fact that $\lambda\in\lambda_0+i\mathbb R^+$,
\[\left|{\rm e}^{ 2s^{3/2}g(\lambda)}\right|={\rm e}^{-\frac{2\sqrt{2}}{3}s^{3/2}|\lambda-\lambda_0|^{3/2}}
\left|
{\rm e}^{\frac{sT^{1/3}}{\pi i}\int_{\lambda_0}^{\lambda}
\left(
\int_{-\infty}^{\lambda_0}\log\frac{(\eta-\lambda_0)^{1/2}-i\sqrt{\lambda_0-\xi}}{(\eta-\lambda_0)^{1/2}+i\sqrt{\lambda_0-\xi}}d\sigma(sT^{1/3}\xi)\right)
d\eta}\right|.\]
It now remains to show that the last factor in the above expression is smaller than $1$, which is true if we can show that 
$$
\Re \left(\int_{-\infty}^{\lambda_0}\log\frac{(\eta-\lambda_0)^{1/2}-i\sqrt{\lambda_0-\xi}}{(\eta-\lambda_0)^{1/2}+i\sqrt{\lambda_0-\xi}}d\sigma(sT^{1/3}\xi) \right) < 0 \quad \forall \eta \in\lambda_0+i\mathbb R^+
$$
This follows from the observation that, for any $\xi\in\mathbb R$ and for any $\eta\in\lambda_0+i\mathbb R^+$,
\[\left|(\eta-\lambda_0)^{1/2}-i\sqrt{\lambda_0-\xi}\right|< \left|(\eta-\lambda_0)^{1/2}+i\sqrt{\lambda_0-\xi}\right|.\]
This completes the proof.
\end{proof}

Now we define, for any $\lambda \notin i\R$ (recall that $\sigma$ has poles on the imaginary line), 
\begin{equation}\label{eq:defw1}
	w(\lambda) := 1+\frac{s^{-1/2}T^{1/3}}{2\pi}\int_{-\infty}^{\lambda_0}\frac{\sigma(sT^{1/3}\xi)-\sigma(sT^{1/3}\lambda)}{\sqrt{\lambda_0-\xi}(\xi-\lambda)}d\xi.
\end{equation}
By \eqref{exprgprime} and the identity theorem, we have
\begin{equation}
	2g'(\lambda) - V'(\lambda) = 2(\lambda - \lambda_0)^{1/2}w(\lambda)
\end{equation}
for $\lambda\in\mathbb C\setminus((-\infty,\lambda_0]\cup i\mathbb R)$.
By taking limits, we obtain for $\lambda<\lambda_0$ that 
\begin{equation}\label{equationgVp}
	2g'_+(\lambda) - V'(\lambda) = 2i \sqrt{\lambda_0 - \lambda}w(\lambda).
\end{equation}
We also denote
\begin{equation}\label{def:psi}
	\psi(\lambda) := 2\sqrt{\lambda_0-\lambda}\,w(\lambda) = -2ig_+'(\lambda)+iV'(\lambda),\qquad \lambda \in (-\infty,\lambda_0),
\end{equation}
and observe that, by 
\eqref{exprg} and \eqref{eq:jumpgprime}, we have
\begin{equation}\label{exprpsi}
\psi(\lambda) =2\sqrt{\lambda_0 - \lambda}\left(1+\frac{s^{-1/2}T^{1/3}}{2\pi}P.V.\int_{-\infty}^{\lambda_0}
\frac{\sigma(sT^{1/3}\xi)}{\sqrt{\lambda_0-\xi}}\frac{d\xi}{\xi-\lambda}\right),
\end{equation}
where we denote with $P.V.$ the Cauchy principal value integral.
By \eqref{exprg2}, we can also write
\begin{equation}\label{exprpsi2}
\psi(\lambda)=2\sqrt{\lambda_0-\lambda}+\frac{s^{-1/2}T^{1/3}}{\pi}\int_{-\infty}^{\lambda_0}
\log\frac{\sqrt{\lambda_0-\lambda}+
\sqrt{\lambda_0-\xi}}{\left|\sqrt{\lambda_0-\lambda}-
\sqrt{\lambda_0-\xi}\right|}
d\sigma(sT^{1/3}\xi).
\end{equation}
It is straightforward to verify that $\psi(\lambda)\geq 2\sqrt{\lambda_0-\lambda}$ for any $s,T>0$ and for any $\lambda<\lambda_0$.

\begin{remark}
If $s^{-1/2}T^{1/3}\to 0$, $\lambda_0\to 1$ by \eqref{eq:endpointequation}, and using \eqref{exprpsi} we have that
 $\frac{\sqrt{s}}{2\pi}\psi(\lambda)$ can be approximated by $\frac{\sqrt{s}}{\pi}\sqrt{1-\lambda}$ or, in terms of $\zeta=s\lambda-s$, by $\frac{1}{\pi}\sqrt{-\zeta}$, which is the limiting density of the Airy point process given in \eqref{eq:Airydensity}.
For general $s,T$, in analogy to the Coulomb gas picture from \cite{CGKLDT}, one can heuristically interpret $\frac{\sqrt{s}}{2\pi}\psi(\lambda)$ as an approximation of the density of the points in the Airy point process, but influenced by an additional force which pushes the particles to the left, and which becomes stronger when $T$ increases.
Although we do not have a proof of this fact, we expect that $\frac{\sqrt{s}}{2\pi}\psi(\lambda)$ can be interpreted as the limiting density of the Airy point process, conditioned on the event that the thinned Airy point process (recall the discussion below \eqref{eq:BorodinGorin}) is empty.
\end{remark}

From the representation \eqref{exprpsi2}, we are able to derive the following estimates which will be needed later on.

\begin{proposition}
\label{prop:estimatespsi0}
For any $M>0$, as $s\to\infty$, we have uniformly in $M^{-1}\leq T\leq Ms^{3/2}$ that
\begin{equation}\label{eq:psiest1}
\psi(0)=2\sqrt{\lambda_0}+\frac{s^{-1/2}T^{1/3}}{\pi}\log(sT^{1/3})+\mathcal O(s^{-1/2}T^{1/3}),
\end{equation}
and with the same uniformity, but in addition uniformly in $\lambda\leq \lambda_0$,
\begin{equation}\label{eq:psiest2}
\psi(\lambda)-\psi(0)=2\sqrt{\lambda_0-\lambda}-2\sqrt{\lambda_0}+\mathcal O(s^{-1/2}T^{1/3}\lambda)+\mathcal O_{\rm sym}\left(s^{-1/2}T^{1/3}\log\left(2+ sT^{1/3}|\lambda|)\right)\right),
\end{equation}
where $\mathcal O_{\rm sym}\left(s^{-1/2}T^{1/3}\log\left(2+ sT^{1/3}|\lambda|)\right)\right)$ denotes a function of order $\mathcal O\left(s^{-1/2}T^{1/3}\log\left(2+ sT^{1/3}|\lambda|)\right)\right)$ which is moreover symmetric as a function of $\lambda$.
\end{proposition}
\begin{proof}
We obtain from \eqref{exprpsi2} that
\begin{multline*}\psi(0)=2\sqrt{\lambda_0}+\frac{2s^{-1/2}T^{1/3}}{\pi}\int_{-\infty}^{\lambda_0}\log\left(\sqrt{\lambda_0}+\sqrt{\lambda_0-\xi}\right)d\sigma(sT^{1/3}\xi)\\
-\frac{s^{-1/2}T^{1/3}}{\pi}\int_{-\infty}^{\lambda_0}\log|\xi|d\sigma(sT^{1/3}\xi).
\end{multline*}
The second term at the right hand side is uniformly $\mathcal O(s^{-1/2}T^{1/3})$, and the last term is equal to 
\begin{multline*}-\frac{s^{-1/2}T^{1/3}}{\pi}\int_{-\infty}^{sT^{1/3}\lambda_0}\log|u|d\sigma(u)
+\frac{s^{-1/2}T^{1/3}}{\pi}\log(sT^{1/3})\int_{-\infty}^{\lambda_0}d\sigma(sT^{1/3}\xi)
\\=\frac{s^{-1/2}T^{1/3}}{\pi}\log(sT^{1/3})+\mathcal O(s^{-1/2}T^{1/3}).\end{multline*}
This yields \eqref{eq:psiest1}.

For \eqref{eq:psiest2}, we have after a straightforward computation that
\begin{multline*}\psi(\lambda)-\psi(0)=2\sqrt{\lambda_0-\lambda}-2\sqrt{\lambda_0}+\frac{2s^{-1/2}T^{1/3}}{\pi}\int_{-\infty}^{\lambda_0}\log\frac{\sqrt{\lambda_0-\lambda}+\sqrt{\lambda_0-\xi}}{\sqrt{\lambda_0}+\sqrt{\lambda_0-\xi}}d\sigma(sT^{1/3}\xi)\\
-\frac{s^{-1/2}T^{1/3}}{\pi}\int_{-\infty}^{\lambda_0}\log\left|\frac{\xi-\lambda}{\xi}\right|d\sigma(sT^{1/3}\xi).
\end{multline*}
It is easy to verify that the integral on the first line gives a uniform $\mathcal O(s^{-1/2}T^{1/3}\lambda)$ contribution. Up to an exponentially small term as $s\to\infty$ coming from the integration over $(-\infty,-\lambda_0)$, the term on the second line can be written as
\[-\frac{s^{-1/2}T^{1/3}}{\pi}\int_{-\lambda_0}^{\lambda_0}\log\left|\frac{\xi-\lambda}{\xi}\right|d\sigma(sT^{1/3}\xi)=-\frac{s^{-1/2}T^{1/3}}{\pi}\int_{-sT^{1/3}\lambda_0}^{sT^{1/3}\lambda_0}\log\left|\frac{u-sT^{1/3}\lambda}{u}\right|\sigma'(u)du.\]
This is a symmetric function of $\lambda$, since $\sigma'$ is symmetric.
For $\lambda=\mathcal O(s^{-1}T^{-1/3})$, this is of order $\mathcal O(s^{-1/2}T^{1/3})$. As $sT^{1/3}|\lambda|\to\infty$, it can be estimated as $\mathcal O\left((s^{-1/2}T^{1/3}\log(sT^{1/3}|\lambda|)\right)$, and the second estimate \eqref{eq:psiest2} follows from these observations.
\end{proof}

\begin{proposition}\label{prop:hatpsi}
For $s$ sufficiently large and for $M^{-1}\leq T\leq Ms^{3/2}$ with $M>0$ arbitrary, there exists $\epsilon>0$ such that $w$ is analytic for $|\lambda-\lambda_0|<\epsilon$. 
Moreover, $w(\lambda_0)\geq 1$ and
\[w(\lambda)=w(\lambda_0)+\mathcal O(\lambda-\lambda_0),\] uniformly in $M^{-1}\leq T\leq Ms^{3/2}$ and $|\lambda-\lambda_0|<\epsilon$ as $s\to\infty$.
\end{proposition}
\begin{proof}
The analyticity of $w(\lambda)$ in a neighborhood of $\lambda_0$ follows directly from the fact that $\lambda_0(s,T)\geq \kappa >0$ for  $s$ sufficiently large and $M^{-1}\leq T\leq Ms^{3/2}$ (by Proposition \ref{prop:asendpoint}) and it suffices to choose $\epsilon>0$ sufficiently small, such that the disk centered at $\lambda_0$ with radius $\epsilon$ does not intersect the imaginary axis.
It is also easy to see from the definition \eqref{eq:defw1} that
$w(\lambda)$ is real-valued for $\lambda\in(\lambda_0-\epsilon,\lambda_0+\epsilon)$, and that $w(\lambda_0)\geq 1$.

To prove that $w(\lambda)=w(\lambda_0)+\mathcal O(\lambda-\lambda_0)$ uniformly in $T$ and for $|\lambda-\lambda_0|<\epsilon$ as $s\to\infty$, it is sufficient to show that $w'(\lambda)=\mathcal O(1)$
uniformly.
To see this, we write
\[w'(\lambda)=1+\frac{s^{-1/2}T^{1/3}}{2\pi}\int_{-\infty}^{\lambda_0}\partial_\lambda h(\xi;\lambda)\frac{d\xi}{\sqrt{\lambda_0-\xi}}
\quad \mbox{ with }\quad h(\xi;\lambda)=\frac{\sigma(sT^{1/3}\xi)-\sigma(sT^{1/3}\lambda)}{\xi-\lambda}.\]
Separating the contributions to the integral coming from a small neighbourhood of $\lambda_0$ and elsewhere, we obtain 
\[
w'(\lambda)=1+\frac{s^{-1/2}T^{1/3}}{2\pi}\int_{-\infty}^{\lambda_0-2\epsilon}\partial_\lambda h(\xi;\lambda)\frac{d\xi}{\sqrt{\lambda_0-\xi}}+\frac{s^{-1/2}T^{1/3}}{2\pi}\int_{\lambda_0-2\epsilon}^{\lambda_0}\partial_\lambda h(\xi;\lambda)\frac{d\xi}{\sqrt{\lambda_0-\xi}}.
\]
For $\xi<\lambda_0-2\epsilon$, we have
\[
\left|\partial_\lambda h(\xi;\lambda)\right|\leq\left|\frac{\sigma(sT^{1/3}\xi)-\sigma(sT^{1/3}\lambda)}{(\xi-\lambda)^2}\right|+sT^{1/3}\left|\frac{\sigma'(sT^{1/3}\lambda)}{\xi-\lambda}\right|=\mathcal O\left(\frac{1}{|\xi| +1}\right),
\]
uniformly in $\xi, \lambda, T$, by the boundedness of $\sigma$ and the exponential decay of $\sigma'$.
For $\lambda_0-2\epsilon\leq \xi\leq \lambda_0$ on the other hand, Taylor expanding $\sigma(sT^{1/3}\xi)$ around $\lambda$ yields
\[
\left|\partial_\lambda h(\xi;\lambda)\right|=\left|\frac{\sigma(sT^{1/3}\xi)-\sigma(sT^{1/3}\lambda)-
sT^{1/3}\sigma'(sT^{1/3}\lambda)(\xi-\lambda)}{(\xi-\lambda^2)}\right|\leq s^2T^{2/3} \frac{\max_{\lambda_0-2\epsilon\leq \xi\leq \lambda_0}|\sigma''(sT^{1/3}\xi)|}{2}
\]
and this decays exponentially fast in $s$ if $\epsilon$ is sufficiently small.
Substituting these two estimates in the above expression for $w'(\lambda)$ and using the fact that $T\leq Ms^{3/2}$, we obtain that $w'(\lambda)$ is uniformly bounded.

\end{proof}

\section{Asymptotic analysis for $\Psi$ as $s\to \infty$}\label{section:RH}
The aim of this section is to obtain asymptotics as $s\to +\infty$ for the RH solution $\Psi$, uniformly in $M^{-1} < T < Ms^{3/2}$ for any $M>0$, via the Deift-Zhou steepest descent method \cite{DZ}. 
To achieve this goal, we will apply two invertible transformations $\Psi \mapsto S \mapsto R$, in order to arrive at a RH problem for $R$ with small jump matrices and normalized such that  $\lim_{\lambda\to\infty}R(\lambda)=I$.
We will then be able to conclude that the RH solution is uniformly in $T$ and $\lambda$ close to $I$ as $s\to\infty$, and by inverting the transformations, we can obtain asymptotics for $\Psi$.
For the first transformation, we will use the $g$-function constructed in the previous section. For the second one, we will need to construct two parametrices, a local parametrix near $\lambda_0$ and a global paramerix elsewhere, which will turn out to be good approximations to $S$.

 \subsection{First transformation $\Psi\mapsto S$}\label{S}
Given $g$ defined as in \eqref{exprg} and \eqref{def:g}, satisfying the properties \eqref{jumpg} and \eqref{as:g}, we define the matrix function $S(\lambda)$ as
\begin{equation}
\label{def:S}
S(\lambda) := {\rm e}^{-\frac{s^{3/2}}{2}V(\lambda_0)\sigma_3}\begin{pmatrix}1&{i\left(g_1-\frac{1}{4}\right)s^{3/2}}\\0&1\end{pmatrix}s^{-\frac{1}{4} \sigma_3}\Psi(s\lambda-s){\rm e}^{s^{3/2}\left(g(\lambda)+\frac{V(\lambda_0)}2\right)\sigma_3}
\end{equation}
where the constant $\zeta_0$ in the RH problem for $\Psi$ is taken such that $s\lambda_0 - s = \zeta_0$.
We will show below that $S$ satisfies the following RH problem.
\subsubsection*{RH problem for $S$}\label{RHS}
\begin{itemize}
\item[(a)] $S : \mathbb C \backslash \widehat\Gamma \rightarrow \mathbb C^{2\times 2}$ is analytic, with
\begin{equation}\label{eq:defGammahat}
\widehat\Gamma := \lambda_0 + (\R \cup i\R),
\end{equation}
oriented in the same way as the contour $\Gamma$ for $\Psi$, i.e.\ the real line is oriented from left to right, the vertical half-lines in the upper and lower half plane are pointing to $\lambda_0$.

\item[(b)] $S(\lambda)$ has continuous boundary values as $\lambda\in\widehat\Gamma\backslash \{\lambda_0\}$ is approached from the left or right and they are related by
\begin{equation}\label{JumpS}
\begin{array}{ll}
S_+(\lambda) = S_-(\lambda) \begin{pmatrix} 1 & 0 \\ {\rm e}^{s^{3/2}(2g(\lambda)-V(\lambda)+V(\lambda_0))} & 1 \end{pmatrix} & \textrm{for } \lambda \in \lambda_0+i\R^{\pm}, \\
S_+(\lambda) = S_-(\lambda) \begin{pmatrix} 0 & 1 \\ -1 & 0 \end{pmatrix} & \textrm{for } \lambda \in \lambda_0 + \R^-, \\
S_+(\lambda) = S_-(\lambda) \begin{pmatrix} 1 & {\rm e}^{-s^{3/2}(2g(\lambda)-V(\lambda)+V(\lambda_0))} \\ 0 & 1 \end{pmatrix} & \textrm{for } \lambda \in \lambda_0 + \R^+.
\end{array}
\end{equation}
\item[(c)] As $\lambda \rightarrow \infty$, $S$ has the asymptotic behavior
\begin{equation}
\label{eq:Sasympinf}
S(\lambda) = \left( I +  \mathcal O(\lambda^{-1})\right) \lambda^{\frac{1}{4} \sigma_3} A^{-1} .
\end{equation}
\item[(d)] As $\lambda\to\lambda_0$, $S$ remains bounded.
\end{itemize}

\begin{proposition}
	The function $S$ defined as in equation \eqref{def:S} satisfies the above RH problem.
\end{proposition}
\proof 
By construction, $S$ is analytic everywhere except on $\widehat\Gamma$.
The jump relations \eqref{JumpS} for $S$ can be inferred from those of $\Psi$ given in \eqref{jumpPsi}. On each of the four half-lines in the jump contour $\widehat\Gamma$, we compute the jump matrices $S_-^{-1}S_+$ using \eqref{def:S}.
In this product, we note that the factors in \eqref{def:S} in front of $\Psi$ are independent of $\lambda$ and thus cancel out, and we can conclude that
\[S_-^{-1}(\lambda)S_+(\lambda)=
{\rm e}^{-s^{3/2}\left(g_-(\lambda)+\frac{V(\lambda_0)}2\right)\sigma_3}
\Psi_-^{-1}(s\lambda-s)\Psi_+(s\lambda-s)
{\rm e}^{s^{3/2}\left(g_+(\lambda)+\frac{V(\lambda_0)}2\right)\sigma_3}\]
for $\lambda\in\widehat\Gamma$.
On $\lambda_0 + {i}\R^\pm$ and on $\lambda_0 +\R^+$,
$g$ is analytic, and the required jump relation for $S$ follows from the jump relations \eqref{jumpPsi} for $\Psi$ and the definition of \eqref{def:h} of $V$. On $\lambda_0 + \R^-$, $g$ is not analytic and we need to use in addition the relation \eqref{jumpg}. The asymptotics for $S$ as $\lambda\to\infty$ follow from the asymptotics for $\Psi$ as $\zeta\to\infty$, see \eqref{eq:psiasympinf}, together with the asymptotics for $g$ given in \eqref{as:g} after a straightforward calculation (note that the factors in front of $\Psi$ in \eqref{def:S} are present precisely to make this calculation work). Finally, it is easy to see that $S$ is bounded near $\lambda_0$, since the same is true for $\Psi(s\lambda-s)$ and for $g$.
\qed\\

We can already make the important observation that, by Proposition \ref{prop:ineqspsi}, the jump matrices for $S$ tend exponentially fast to $I$ as $s\to \infty$, except near $\lambda_0$ and except on $(-\infty,\lambda_0]$, where the jump matrix is constant. This the reason why we have to introduce two different parametrices: a local one to approximate $S$ in the neighborhood of $\lambda_0$ and a global one to approximate $S$ elsewhere. 

\subsection{Global parametrix}\label{P}

If we ignore the small jumps of $S$ and the jumps in the vicinity of $\lambda_0$, we are led to the following RH problem. 

\subsubsection*{RH problem for $P^{\infty}$}\label{RHPInfty}

\begin{itemize}
\item[(a)] $P^{\infty} : \mathbb C \backslash (-\infty,\lambda_0] \rightarrow \mathbb C^{2\times 2}$ is analytic.

\item[(b)] $P^{\infty}$ has continuous boundary values on $(-\infty, \lambda_0)$ and they are related by
\begin{align}
	P^{\infty}_+(\lambda) &= P^{\infty}_-(\lambda) \begin{pmatrix} 0 & 1 \\ -1 & 0 \end{pmatrix}, & \lambda \in (-\infty, \lambda_0).
\end{align}
\item[(c)] As $\lambda \rightarrow \infty$, we have
\begin{equation}\label{asPInfty}
	P^\infty(\lambda) = (I + \mathcal O(\lambda^{-1}))\lambda^{\frac{1}4 \sigma_3} A^{-1}.
\end{equation}
\end{itemize}

Without imposing an additional condition near $\lambda_0$, the solution to this RH problem is not unique. One can also show that it is impossible to construct a solution which, like $S$, remains bounded near $\lambda_0$. This also shows that we cannot expect $P^\infty$ to be a good approximation of $S$ near $\lambda_0$.
The best we can do, is construct a solution $P^\infty(\lambda) = \mathcal O(\lambda - \lambda_0)^{-1/4}$ as $\lambda \rightarrow \lambda_0$, and by imposing this, the solution is moreover unique (but we do not need this). It is easy to verify that $P^{\infty}$ given by
\begin{equation}\label{def:Pinf}
	P^\infty(\lambda) = (\lambda-\lambda_0)^{\frac{1}{4} \sigma_3} A^{-1},
\end{equation}
where $(\lambda-\lambda_0)^{\pm\frac{1}{4}}$ is analytic except on $(-\infty,\lambda_0]$ and positive for $\lambda>\lambda_0$, solves the RH problem.

\subsection{Local Airy parametrix}

We will now construct the local parametrix $P$ in a disk around $\lambda_0$ in such a way that it has exactly the same jumps as $S$ has inside the disk, and in such a way that it is close to $P^\infty$ on the boundary of the disk.
Let us fix a disk $U$ around $\lambda_0$ of radius $\epsilon$, with $\epsilon>0$ small enough such that the result of Proposition \ref{prop:hatpsi} holds. In particular, for $\lambda \in U$ we have 
\begin{align}2g(\lambda)-V(\lambda)+V(\lambda_0)&=2\int_{\lambda_0}^{\lambda}(\eta-\lambda_0)^{1/2}w(\eta)d\eta\nonumber \\
&=
\frac{4}{3}w(\lambda_0)(\lambda-\lambda_0)^{3/2}+2\int_{\lambda_0}^{\lambda}(\eta-\lambda_0)^{1/2}(w(\eta)-w(\lambda_0))d\eta
\nonumber \\&=\frac{4}{3}w(\lambda_0)(\lambda-\lambda_0)^{3/2}+\mathcal O((\lambda-\lambda_0)^{5/2}),\label{eq:boundmu0}
\end{align}
uniformly in $\lambda, T$ as $s\to\infty$. Since $w(\lambda_0) \geq 1$, we can choose $\epsilon$ small enough such that, in addition, 
\begin{equation}\label{eq:mubounded}
	\left|\left(2g(\lambda)-V(\lambda)+V(\lambda_0)\right)^{2/3}\right|\geq c\epsilon,\quad\mbox{for $|\lambda-\lambda_0|=\epsilon$}
\end{equation}
and for some constant $c$ independent of $s,T$ and $\lambda$.\\

From the above expansion, it follows that we can define a conformal mapping $\mu : U \rightarrow \mathbb C$ by imposing
\begin{equation}\label{def:mu}
	\frac{2}3\mu^{3/2}(\lambda) = g(\lambda)-\frac{V(\lambda)}2+\frac{V(\lambda_0)}2
\end{equation}
and $\mu'(\lambda_0)>0$.
Now we can define the local parametrix $P(\lambda)$ for $\lambda \in U$ by
\begin{equation}\label{def:P}
P(\lambda) := \left(\frac{\lambda - \lambda_0}{s\mu(\lambda)}\right)^{\frac{\sigma_3}4} \Phi_k(s\mu(\lambda)){\rm e}^{s^{3/2}\left(g(\lambda)-\frac{V(\lambda)}{2} + \frac{V(\lambda_0)}2\right)\sigma_3},
\end{equation}
where $\Phi_k$, $k={\rm I}, {\rm II}, {\rm III}, {\rm IV}$ are the functions constructed in terms of the Airy function given in \eqref{eqA}--\eqref{eqD}, and where we set $k={\rm I}$ for $\Re\mu(\lambda)>0$, $\Im \mu(\lambda)>0$, $k={\rm II}$ for $\Re\mu(\lambda)<0$, $\Im \mu(\lambda)>0$, $k={\rm III}$ for $\Re\mu(\lambda)<0$, $\Im \mu(\lambda)<0$, and $k={\rm IV}$ for $\Re\mu(\lambda)>0$, $\Im \mu(\lambda)<0$.
\begin{proposition}\label{prop:P}\hfill
\begin{enumerate}\item $P$ satisfies the same jump relations as $S$ on $\widehat\Gamma\cap U$, given in \eqref{JumpS}.
\item For $\lambda \in \partial U$, we have uniformly in $\lambda$ and $T$ that
	\begin{equation}
		P(\lambda)P^{\infty}(\lambda)^{-1} = I+\mathcal O(s^{-3/2}) \quad \text{as} \quad s\rightarrow \infty.
	\end{equation}
	\end{enumerate}
\end{proposition}
\begin{proof}
\begin{enumerate}
\item We can compute the jump matrices for $P$ by computing
$P_-^{-1}P_+$. A straightforward calculation using the relations (similar to \eqref{eq:jump1}--\eqref{eq:jump3}) 
\[
\Phi_{\rm I} = \Phi_{\rm II} \begin{pmatrix} 1 & 0 \\ 1 & 1 \end{pmatrix},\quad 
\Phi_{\rm II} = \Phi_{\rm III} \begin{pmatrix} 0 & 1 \\ -1 & 0 \end{pmatrix},\quad \Phi_{\rm III} = \Phi_{\rm IV} \begin{pmatrix} 1 & 0 \\ 1 & 1 \end{pmatrix},\quad 
\Phi_{\rm I} = \Phi_{\rm IV} \begin{pmatrix} 1 & 1 \\ 0 & 1 \end{pmatrix},\]
and the relation \eqref{jumpg} for $g$ shows that $P$ has exactly the same jump matrices as $S$ on $\widehat\Gamma\cap U$.
\item
We first observe that, by \eqref{eq:mubounded}, $|\mu(\lambda)|$ is bounded from below by a positive constant independent of $\lambda, s, T$ for $\lambda\in\partial U$, such that $s\mu(\lambda)$ is large and we can substitute the asymptotics for $\Phi_{\rm I}$, $\Phi_{\rm II}$, $\Phi_{\rm III}$, $\Phi_{\rm IV}$ (like the ones given for $\Phi$ in \eqref{asPhi}) in \eqref{def:P}. This gives
\[
	P(\lambda)P^{\infty}(\lambda)^{-1}= (\lambda - \lambda_0)^{\frac{\sigma_3}4}A^{-1}\left(I  + \mathcal O(s^{-3/2}) \right) A (\lambda - \lambda_0)^{-\frac{\sigma_3}4} 
\]
as $s\to\infty$, uniformly in $\lambda\in\partial U$ and $T$, and we obtain the result.
\end{enumerate}
\end{proof}

\subsection{Second transformation and small-norm RH problem}\label{R}

We now define
\begin{equation}\label{def:R}
R(\lambda) := \begin{cases}
S(\lambda)P(\lambda)^{-1}&\mbox{for $\lambda\in U$,}\\
S(\lambda)P^\infty(\lambda)^{-1}&\mbox{elsewhere,}
\end{cases}
\end{equation}
and show that $R$ solves the following RH problem.

\subsubsection*{RH problem for R}

\begin{itemize}
	\item[(a)] $R : \mathbb C \backslash \Gamma_R \longrightarrow \mathbb C^{2 \times 2}$ is analytic, with
	$$ \Gamma_R := \partial U \cup \left((\lambda_0+(i\R\cup\R^+))\backslash U\right).$$
	\item[(b)] $R$ has jump relations
					$$R_+(\lambda) = R_-(\lambda)J_R(\lambda), \quad \lambda \in \Gamma_R,$$
			where $J_R$ takes the form
	\begin{equation}\label{JR}
		J_R(\lambda) = \begin{cases}
						I + \mathcal O\left(\frac{{\rm e}^{-\eta s^{3/2}}}{|\lambda|^2 + 1} \right), &\lambda \in \Big( \lambda_0 \pm i(\epsilon,+\infty) \Big) \cup (\lambda_0 + \epsilon,+\infty), \\[1ex]
					
						I + \mathcal O(s^{-3/2}), &\lambda \in \partial U,
					\end{cases}
	\end{equation}
	as $s\to\infty$, uniformly in $\lambda, T$, for some
	$\eta>0$.
	\item[(c)] $R(\lambda) = I + \mathcal O(\lambda^{-1}) \quad \text{as} \quad \lambda \rightarrow \infty.$ 			
\end{itemize}
\begin{proof}Condition (c) is easy to verify from the asymptotics of $S$ and $P^\infty$.
Next, since $S$ and $P$ have exactly the same jump relations inside $U$ and since $S$ and $P^\infty$ have exactly the same jump relations on $(-\infty,\lambda_0-\epsilon)$, $R$ is analytic on those parts of the jump contour $\widehat\Gamma$.
	On the rays  $\Big( \lambda_0 \pm i(\epsilon,+\infty) \Big)$, the jump matrices for $R$ can be computed using the jump relations for $S$ as
\[J_R(\lambda)=P^\infty(\lambda)S_-^{-1}(\lambda)S_+(\lambda)P^\infty(\lambda)^{-1}=P^\infty(\lambda)\begin{pmatrix} 1 & 0 \\ {\rm e}^{s^{3/2}(2g(\lambda)-V(\lambda)+V(\lambda_0))} & 1 \end{pmatrix}P^\infty(\lambda)^{-1}.\]	
Noting that $P^\infty(\lambda)^{-1}$ and $P^\infty(\lambda)$ are uniformly $\mathcal O(\lambda^{\sigma_3/4})$, and using \eqref{ineq2}, we obtain the first estimate in \eqref{JR}.
A similar computation yields the result on  $(\lambda_0 + \epsilon,+\infty)$, using equation \eqref{ineq1} instead of \eqref{ineq2}.
	On $\partial U$, the jump matrix for $R$ is equal to $P(P^\infty)^{-1}$, and the required estimate follows from Proposition \ref{prop:P}	. 
\end{proof}

\begin{corollary}\label{cor:asR}
	As $s \rightarrow \infty$, $R$ and $R'$ have the asymptotics 
	\begin{equation}\label{as:R}
		R(\lambda) = I + \mathcal O\left(\frac{1}{s^{3/2}(|\lambda|+1)}\right), \quad \quad R'(\lambda) = \mathcal O\left(\frac{1}{s^{3/2}(|\lambda|^2+1)}\right),
	\end{equation}
	uniformly in $\lambda\in\mathbb C\setminus\Gamma_R$ and uniformly in $M^{-1}\leq T\leq Ms^{3/2}$ for any $M>0$.
\end{corollary}
\begin{proof}
The estimate for $R$ follows from the RH problem for $R$ and standard small-norm arguments in the theory of RH problems, see e.g.\ \cite{Deift, DKMVZ}.
For the derivative of $R$, by Cauchy's formula, we can write for any $\lambda\in\mathbb C\setminus\Gamma_R$,
$$
R'(\lambda) = \int_{\gamma} \frac{R(\xi)}{(\xi - \lambda)^2} \frac{d \xi}{2 \pi i},
$$
where $\gamma$ is a small counterclockwise oriented circle of radius $\rho$ around $\lambda$ whose interior lies in the domain of analyticity for $R$. 
Substituting the asymptotics for $R$, we obtain
\[R'(\lambda) = \mathcal O\left(\frac{1}{\rho s^{3/2}(|\lambda|^2+1)}\right).\]
For $\lambda$ not too close to $\Gamma_R$, say ${\rm dist}(\lambda,\Gamma_R)\geq \eta$ for some small $\eta>0$, this yields the estimate since we can take $\rho=\eta$. For ${\rm dist}(\lambda,\Gamma_R)<\eta$, we need to deform the jump contour $\Gamma_R$. As is common in RH analysis, we can do this by analytically continuing $R$ (to $\widetilde R$) across the contour, and in this way we  can obtain a deformed contour $\widetilde\Gamma_R$ and a circle $\gamma$ of radius $\eta$ whose interior lies in $\mathbb C\setminus\widetilde\Gamma_R$, such that
\[R'(\lambda) = \int_{\gamma} \frac{\widetilde R(\xi)}{(\xi - \lambda)^2} \frac{d \xi}{2 \pi i},\qquad \widetilde R(\lambda) = I + \mathcal O\left(\frac{1}{s^{3/2}(|\lambda|+1)}\right)\]
as $s\to\infty$, uniformly in $\lambda$ and $T$.
We now obtain the same estimate in a similar way as before.
\end{proof}	

Since we know the asymptotics for $R$, we can now obtain asymptotics for $\Psi$ by inverting the transformations $S\mapsto R$ and $\Psi\mapsto S$.

\section{Asymptotics for $Q(s,T)$ as $s\to\infty$}

\subsection{Logarithmic derivative with respect to $T$}

The main result of this subsection is Proposition \ref{prop:logders2}, giving explicit asymptotics for $\partial_T\log Q(s,T)$ as $s\to\infty$, uniformly in $M^{-1} \leq T \leq Ms^{3/2}$, up to terms of order smaller than or equal to $sT^{-2/3}$. As a first step, we prove the following. 
 
 \begin{proposition}\label{prop:logder}
As $s \to \infty$, uniformly in $M^{-1}\leq T\leq Ms^{3/2}$, we have
\begin{align}
&\partial_T\log Q(s,T)\nonumber\\
&\quad =
-\frac{s^{5/2}}{6\pi T^{2/3}}\int_{-\infty}^{\lambda_0}\lambda\sigma(sT^{1/3}\lambda)
\psi(\lambda) d\lambda 
+\frac{s}{12\pi T^{2/3}}\int_{-\infty}^{\lambda_0-\epsilon}\lambda\frac{\cos\left(
s^{3/2}\int_{\lambda_0}^\lambda\psi(\xi)d\xi\right)}{\lambda-\lambda_0}\sigma(sT^{1/3}\lambda)
d\lambda
\nonumber\\ 
&\qquad -\frac{s^2T^{-2/3}}{3}
\int_{\lambda_0-\epsilon}^{\lambda_0+\epsilon}
\lambda \sigma(sT^{1/3}\lambda)
\left(K^{\Ai}(s\mu(\lambda),s\mu(\lambda))-
\frac{1}{\pi}s^{1/2}|\mu(\lambda)|^{1/2}{\bf 1}_{(-\infty,\lambda_0]}(\lambda)\right)
\mu'(\lambda) d\lambda \nonumber \\ 
&\qquad 
+\mathcal O(T^{-2/3})
.
\label{diffidtotal}
\end{align}
\end{proposition}

\begin{proof}
We recall from Theorem \ref{diffidentity} that
\[\partial_T\log Q(s,T)=-\frac{1}{6\pi iT^{2/3}}\int_{\mathbb R} (\zeta+s) \sigma'(T^{1/3}(\zeta+s))\left(\widehat\Psi^{-1}\widehat\Psi_\zeta\right)_{21}(\zeta)d\zeta,\]
where $\widehat\Psi$ is the boundary value of the analytic continuation of $\Psi$ from the sector $0<\arg(\zeta-\zeta_0)<\pi/2$ to the upper half plane, given by \eqref{def:hatPsi}.
The explicit invertible transformations $\Psi\mapsto S\mapsto R$ and the (uniform in $\lambda$ and in $T$) asymptotics \eqref{as:R} for $R$, Corollary \ref{cor:asR}, allow us to compute asymptotics for the right hand side of the above equation.

Using \eqref{def:S} and the change of variables $\lambda=\frac{\zeta}{s}+1$, we express $\partial_T\log Q(s,T)$ in terms of $\widehat S(\lambda)$, defined as 
\begin{equation}\label{def:Shat}\widehat S(\lambda)=\begin{cases}S_+(\lambda)&\mbox{for $\lambda>\lambda_0$,}\\
S_+(\lambda)\begin{pmatrix}1&0\\  {\rm e}^{s^{3/2}(2g_+(\lambda)-V(\lambda)+V(\lambda_0))}&1\end{pmatrix}&\mbox{for $\lambda<\lambda_0$.}
\end{cases}\end{equation}
In other words, $\widehat S$ is the boundary value of the analytic continuation of $S$ from the sector $0<\arg(\lambda-\lambda_0)<\pi/2$ to the whole upper half plane. More precisely, we have
\begin{equation}\label{eq:ShatQ1}
	\partial_T\log Q(s,T)=-\frac{s}{6\pi iT^{2/3}}\int_{\mathbb R}\lambda\sigma'(sT^{1/3}\lambda)\left(\widehat S^{-1}\widehat S'\right)_{21}(\lambda){\rm e}^{-s^{3/2}(2g_+(\lambda)+V(\lambda_0))} d\lambda.
\end{equation}
By the definition \eqref{def:h} of $V$ and the definition \eqref{eq:Fredholmid1} of $\sigma$ we easily obtain that $$\sigma'(sT^{1/3}\lambda) = \sigma(sT^{1/3}\lambda){\rm e}^{s^{3/2}V(\lambda)},$$ so that we can write \eqref{eq:ShatQ1} above as
\[\partial_T\log Q(s,T)=-\frac{s}{6\pi iT^{2/3}}\int_{\mathbb R}\lambda\sigma(sT^{1/3}\lambda)\left(\widehat S^{-1}\widehat S'\right)_{21}(\lambda){\rm e}^{-s^{3/2}(2g_+(\lambda)-V(\lambda)+V(\lambda_0))} d\lambda.\]
Next, we split this integral in $3$ parts,
\[\partial_T\log Q(s,T)=I_1+I_2+I_3,\]
where
\begin{equation}\label{def:Ij}I_j=-
\frac{s}{6\pi iT^{2/3}}\int_{A_j} \lambda\sigma(sT^{1/3}\lambda)
\left(\widehat S^{-1}\widehat S'\right)_{21}(\lambda){\rm e}^{-s^{3/2}(2g_+(\lambda)-V(\lambda)+V(\lambda_0))} d\lambda,
\end{equation}
with \[A_1=(-\infty,\lambda_0-\epsilon],\quad A_2=(\lambda_0-\epsilon,\lambda_0+\epsilon),\quad
A_3=[\lambda_0+\epsilon,+\infty).\]

Now, we will compute asymptotics for $I_1, I_2, I_3$, which need to be uniform in $M^{-1}\leq T\leq Ms^{3/2}$, for any $M>0$, as $s\to +\infty$.
Whenever we write that error terms are uniform in $T$ below, we mean that they are uniform with respect to $T$ in this domain. When uniformity in $\lambda$ is also required, we will specify this.

\paragraph{Asymptotics for $I_3$.}
By \eqref{def:R}, we can express
$S$ (and hence $\widehat{S}$) in terms of $P^\infty$ and $R$, for $\lambda \in (\lambda_0 + \epsilon, \infty)$. We have
\begin{align}
\left(\widehat S^{-1}\widehat S'\right)_{21}(\lambda)&=\left((P^{\infty})^{-1}(P^{\infty})'\right)_{21}(\lambda)+\left((P^{\infty})^{-1}R^{-1}R'P^{\infty}\right)_{21}(\lambda)\nonumber\\
&=\frac{i}{4(\lambda-\lambda_0)}+\mathcal O(s^{-3/2})=\mathcal O(1),\label{asS1}
\end{align}
uniformly in $T$ and in $\lambda>\lambda_0+\epsilon$ as $s\to \infty$, where we used  the definition \eqref{def:Pinf} of $P^{\infty}$ and the asymptotics of $R(\lambda)$ (Corollary \ref{cor:asR}) for the second equality. 
Substituting this in \eqref{def:Ij} and using \eqref{ineq1}, we obtain
the bound
\begin{equation}\label{I3}
	I_3=\mathcal O\left(sT^{-2/3}\int_{\lambda_0+\epsilon}^{+\infty} |\lambda| {\rm e}^{-\frac{4}{3}s^{3/2}(\lambda-\lambda_0)^{3/2}}d\lambda\right)
=\mathcal O(s^{-1/2}T^{-2/3}),
\end{equation} 
uniformly in $T$ as $s\to \infty$. This bound largely overestimates the contribution of $I_3$, but it will be sufficient for our needs.

\paragraph{Asymptotics for $I_1$.}
We again express $\widehat S$ in terms of $P^\infty$ and $R$, but now for $\lambda \in (-\infty,\lambda_0 - \epsilon)$. Using \eqref{def:Shat}, \eqref{def:Pinf}, and \eqref{def:psi}, we have
\begin{align}
&\left(\widehat S^{-1}\widehat S'\right)_{21}(\lambda)=is^{3/2}\psi(\lambda){\rm e}^{s^{3/2}(2 g_+(\lambda)-V(\lambda)+V(\lambda_0))}\nonumber\\
&\quad +\left(\begin{pmatrix}1&0\\ {\rm e}^{s^{3/2}(2g_+-V+V(\lambda_0))}&1\end{pmatrix}(P^{\infty})^{-1}R^{-1}R'P^{\infty}\begin{pmatrix}1&0\\- {\rm e}^{s^{3/2}(2g_+-V+V(\lambda_0))}&1\end{pmatrix}\right)_{21}(\lambda)\nonumber\\
&\quad+\left(\begin{pmatrix}1&0\\ {\rm e}^{s^{3/2}(2g_+-V+V(\lambda_0))}&1\end{pmatrix}(P^{\infty})^{-1}(P^{\infty})'\begin{pmatrix}1&0\\- {\rm e}^{s^{3/2}(2g_+-V+V(\lambda_0))}&1\end{pmatrix}\right)_{21}(\lambda)\nonumber\\
&=i
\left(s^{3/2}\psi(\lambda)-\frac{\cos\left(
s^{3/2}\int_{\lambda_0}^\lambda\psi(\xi)d\xi\right)}{2(\lambda-\lambda_0)}
\right)
{\rm e}^{s^{3/2}(2g_+(\lambda)-V(\lambda)+V(\lambda_0))}
 +\mathcal O\left(s^{-3/2}\right),\label{asS2}
\end{align}
where this last error term is uniform in $T$ as $s\to\infty$ and also in $\lambda<\lambda_0-\epsilon$. In the above computation, we used the fact that $2g_+-V+V(\lambda_0)$ is purely imaginary, which follows from \eqref{def:g} and \eqref{equationgVp}.
Substituting this in \eqref{def:Ij}, we obtain
\begin{multline}\label{I1beforelemma}
I_1=-\frac{s^{5/2}}{6\pi T^{2/3}}\int_{-\infty}^{\lambda_0-\epsilon}\sigma(sT^{1/3}\lambda)
\psi(\lambda)\lambda d\lambda 
\\+\frac{s}{12\pi T^{2/3}}\int_{-\infty}^{\lambda_0-\epsilon}
\frac{\cos\left(
s^{3/2}\int_{\lambda_0}^\lambda\psi(\xi)d\xi\right)}{\lambda-\lambda_0}\sigma(sT^{1/3}\lambda)\lambda
d\lambda
+\mathcal O\left(s^{-1/2}T^{-2/3}\right),
\end{multline}
uniformly in $T$ as $s\to\infty$.

\paragraph{Asymptotics for $I_2$.}
For $\lambda\in(\lambda_0-\epsilon,\lambda_0+\epsilon)$, we express $\widehat S$ in terms of the local parametrix $P$ and $R$. If we write $\widehat P$ for the analytic continuation of $P$ from the sector $0<\arg(\lambda-\lambda_0)<\pi/2$, 
i.e.\ 
\begin{equation}\label{def:Phat}\widehat P(\lambda)=\begin{cases}P_+(\lambda)&\mbox{for $\lambda>\lambda_0$,}\\
P_+(\lambda)\begin{pmatrix}1&0\\ {\rm e}^{s^{3/2}(2g_+(\lambda)-V(\lambda)+V(\lambda_0))}&1\end{pmatrix}&\mbox{for $\lambda<\lambda_0$,}
\end{cases}\end{equation}
we have by \eqref{def:Shat} and \eqref{def:R} that
\[
\left(\widehat S^{-1}\widehat S'\right)_{21}(\lambda)=\left(\widehat P^{-1}\widehat P'\right)_{21}(\lambda)+\left(\widehat P^{-1}R^{-1}R'\widehat P\right)_{21}(\lambda).
\]
From the construction of the local parametrix, in particular formula \eqref{def:P}, we can verify that
 $\widehat P(\lambda)=\mathcal O(s^{1/4})$ uniformly in $\lambda$ and in $T$ as $s\to\infty$, which yields (using also the explicit expression of the Airy parametrix, see \eqref{eqA})
\begin{eqnarray*}
\left(\widehat S^{-1}\widehat S'\right)_{21}(\lambda)=2\pi s\mu'(\lambda) {\rm e}^{s^{3/2}(2g_+(\lambda)-V(\lambda)+V(\lambda_0))}\left(\begin{pmatrix}-i\Ai'&\Ai\end{pmatrix}\begin{pmatrix}\Ai'\\i\Ai''\end{pmatrix}\right)_{21}(s\mu(\lambda))
+\mathcal O(s^{-1})
,\end{eqnarray*}
with error term again uniform in $\lambda, T$ as $s\to\infty$. In other words,
\begin{equation}
\begin{array}{lll}
\left(\widehat{S}^{-1}\widehat{S}'\right)_{21}(\lambda) &=& 2\pi i s\mu'(\lambda) {\rm e}^{s^{3/2}(2g_+(\lambda)-V(\lambda)+V(\lambda_0))}K^{\Ai}(s\mu(\lambda), s\mu(\lambda))
+\mathcal O(s^{-1}),\label{asS3}
\end{array}
\end{equation}
with the same uniformity for the error term, where $K^{\rm Ai}$ is the Airy kernel given by \eqref{Airykernel}.

Consequently, 
\begin{equation}\label{I2}
\begin{array}{lll}
I_2 &=& \ds -\frac{s^2T^{-2/3}}{3}\int_{\lambda_0-\epsilon}^{\lambda_0+\epsilon}
\lambda \sigma(sT^{1/3}\lambda)
\left(K^{\Ai}(s\mu(\lambda),s\mu(\lambda))\mu'(\lambda)
\right) d\lambda 
+ \mathcal O(T^{-2/3}),
\end{array}
\end{equation}
uniformly in $T$ as $s\to\infty$.

\medskip

Summing up the contributions \eqref{I1beforelemma} of $A_1$ and \eqref{I2} of $A_2$ together with the small contribution \eqref{I3} of $A_3$,
we obtain
\begin{align}
&\partial_T\log Q(s,T)\nonumber\\
&\quad =
-\frac{s^{5/2}}{6\pi T^{2/3}}\int_{-\infty}^{\lambda_0}\lambda \sigma(sT^{1/3}\lambda)
\psi(\lambda) d\lambda 
+\frac{s}{12\pi T^{2/3}}\int_{-\infty}^{\lambda_0-\epsilon}\lambda\frac{\cos\left(
s^{3/2}\int_{\lambda_0}^\lambda\psi(\xi)d\xi\right)}{\lambda-\lambda_0}\sigma(sT^{1/3}\lambda)
d\lambda
\nonumber\\ 
&\qquad -\frac{s^2T^{-2/3}}{3}
\int_{\lambda_0-\epsilon}^{\lambda_0+\epsilon}
\lambda \sigma(sT^{1/3}\lambda)
\left(K^{\Ai}(s\mu(\lambda),s\mu(\lambda))\mu'(\lambda)-\frac{1}{2\pi}s^{1/2}\psi(\lambda){\bf 1}_{(-\infty,\lambda_0]}(\lambda)\right) d\lambda \nonumber \\ 
&\qquad 
+\mathcal O(T^{-2/3})
,\label{diffidtotal0}
\end{align}
uniformly in $T$ as $s\to\infty$. Moreover, 
by \eqref{def:mu} and \eqref{def:psi}, we have $\psi(\lambda)=- 2 i\mu(\lambda)^{1/2}\mu'(\lambda)$.
Substituting this in the above formula, we obtain equation \eqref{diffidtotal}.
\end{proof}

We now need to compute asymptotics for each of the three terms appearing at the right hand side of \eqref{diffidtotal}. Before doing that, we prove a technical lemma, which expresses the fact that $\sigma(sT^{1/3}\lambda)$ approximates a step function in a weak sense as $sT^{1/3}\to\infty$.

\begin{lemma}\label{lemma:stepapprox}
Suppose that $F:\mathbb R\to\mathbb R$ is such that \[\int_{\mathbb R}\frac{|F(\xi)|}{\xi^2+1}d\xi\leq c_2<\infty.\]
\begin{enumerate}
\item Suppose that $F$ is Lipschitz continuous at $0$, i.e.\ there exist $c_1,\delta>0$ such that
$|F(\xi)-F(0)|\leq c_1|\xi|$ for $|\xi|<\delta$.
Then, there exist $C, r_0>0$, depending only on $c_1, c_2$, and $F(0)$, such that for any $r\geq r_0$, we have
\[\left|\int_{-\infty}^{+\infty}F(\xi)\sigma(r\xi)d\xi -\int_{0}^{+\infty}F(\xi)d\xi\right|\leq Cr^{-2}.\]
\item Suppose that there exist $c_1,\delta>0$ such that
$|F(\xi)-F(0)|\leq c_1\left|\xi \log|\xi|\right|$ for $|\xi|<\delta$.
Then, there exist $C, r_0>0$, depending only on $c_1, c_2$, and $F(0)$, such that for any $r\geq r_0$, we have
\[\left|\int_{-\infty}^{+\infty}F(\xi)\sigma(r\xi)d\xi -\int_{0}^{+\infty}F(\xi)d\xi\right|\leq C|\log r|\  r^{-2}.\]
\end{enumerate}
\end{lemma}
\begin{proof}
We give the proof of the first assertion, the second can be proved similarly.
Writing \[I:=\int_{-\infty}^{+\infty}F(\xi)\sigma(r\xi)d\xi -\int_{0}^{+\infty}F(\xi)d\xi,\]
we have
\[I=B_1+B_2+B_3,\]
with 
\begin{align*}
&B_1=F(0)\int_{-\infty}^{+\infty}\hat\sigma(r\xi)d\xi,\\
&B_2=\int_{-\delta}^{\delta}\left(F(\xi)-F(0)\right)\hat\sigma(r\xi)d\xi,\\
&B_3=\int_{\mathbb R\setminus[-\delta,\delta]}\left(F(\xi)-F(0)\right)\hat\sigma(r\xi)d\xi,
\end{align*}
where
$\hat\sigma(x)=\sigma(x)-{\bf 1}_{(0,+\infty)}(x)$.
Since $\hat\sigma(x)$ is an antisymmetric function of $x$, we have $B_1=0$.
Next, we can estimate $|B_2|$ as follows using the Lipschitz continuity and \eqref{ineq:sigma},
\[|B_2|\leq c_1\int_{-\delta}^{\delta}|\xi|\ |\hat\sigma(r\xi)|d\xi=c_1r^{-2}\int_{-\delta r}^{\delta r}|\lambda|\ |\hat\sigma(\lambda)|d\lambda\leq \frac{c_1}{r^2}\int_{-\infty}^{+\infty}|\lambda|\ e^{-|\lambda|}d\lambda.\]

Finally, for $B_3$ we have again using \eqref{ineq:sigma} for large enough $r$, 
\[|B_3|\leq \hat\sigma(\pm r\delta)(r^2\delta^2+1)\ \int_{-\infty}^{+\infty}\frac{|F(\xi)-F(0)|}{\xi^2+1}d\xi\leq (c_2+\pi |F(0)|)e^{-\delta r}(\delta^2r^2+1).\]
This proves the required estimate.
\end{proof}

We can now proceed with the analysis of the three terms on the right hand side of \eqref{diffidtotal}.

\begin{lemma}\label{lemma:oscillint}
As $s\to\infty$, uniformly in $M^{-1} \leq T \leq Ms^{3/2} $
\[\frac{s}{12\pi T^{2/3}}\int_{-\infty}^{\lambda_0-\epsilon}\lambda\frac{\cos\left(
s^{3/2}\int_{\lambda_0}^\lambda\psi(\xi)d\xi\right)}{\lambda-\lambda_0}\sigma(sT^{1/3}\lambda) d\lambda
=\mathcal O(s^{-1/2}T^{-2/3}).\]
\end{lemma}
\begin{proof}
We start by rewriting
\begin{multline*}
\frac{s}{12\pi T^{2/3} }\int_{-\infty}^{\lambda_0-\epsilon}\lambda \frac{\cos\left(
s^{3/2}\int_{\lambda_0}^\lambda\psi(\xi)d\xi\right)}{\lambda-\lambda_0}\sigma(sT^{1/3}\lambda) d\lambda=
\frac{s}{12\pi T^{2/3} }\int_{0}^{\lambda_0-\epsilon}\lambda\frac{\cos\left(
s^{3/2}\int_{\lambda_0}^\lambda\psi(\xi)d\xi\right)}{\lambda-\lambda_0} d\lambda
\\
+\frac{s}{12\pi  T^{2/3}}\int_{-\infty}^{\lambda_0-\epsilon}
\lambda\cos\left(
s^{3/2}\int_{\lambda_0}^\lambda\psi(\xi)d\xi\right)
\frac{(\sigma(sT^{1/3}\lambda)-{\bf 1}_{(0,\lambda_0-\epsilon)}(\lambda)) }{\lambda-\lambda_0}d\lambda.
\end{multline*}
The first term at the right hand side can be written as
\[\frac{1}{12\pi s^{1/2}T^{2/3} }\int_{0}^{\lambda_0-\epsilon}\lambda\frac{s^{3/2}\psi(\lambda)
\cos\left(
s^{3/2}\int_{\lambda_0}^\lambda\psi(\xi)d\xi\right)}{(\lambda-\lambda_0)\psi(\lambda)} d\lambda,
\]
and bounded uniformly in $T$ by $\mathcal O(s^{-1/2}T^{-2/3})$ as $s\to\infty$ after integration by parts and by recalling that $\psi(\lambda)\geq 2\sqrt{\lambda_0-\lambda}$.

Next, we use \eqref{ineq:sigma} to bound the second term in absolute value by 
\[\frac{s}{12\pi\epsilon T^{2/3}}\int_{-\infty}^{\lambda_0-\epsilon}
|\lambda| {\rm e}^{-sT^{1/3}|\lambda|}d\lambda
=
\frac{1}{12\pi\epsilon  s T^{4/3}}\int_{-\infty}^{+\infty}
{\rm e}^{-|u|} |u| d u=\mathcal O(s^{-1}T^{-4/3})
\]
uniformly in $T$ as $s\to\infty$. These two estimates yield the result.
\end{proof}

\begin{lemma}\label{lemma:Airyterm}
As $s\to\infty$, we have uniformly in $T$,
\begin{multline}\label{eq:boundAiry}-\frac{s^2T^{-2/3}}{3}
\int_{\lambda_0-\epsilon}^{\lambda_0+\epsilon}
\sigma(sT^{1/3}\lambda)
\left(K^{\Ai}(s\mu(\lambda),s\mu(\lambda))-
\frac{1}{\pi}s^{1/2}|\mu(\lambda)|^{1/2}{\bf 1}_{(-\infty,\lambda_0]}(\lambda)\right) \lambda
\mu'(\lambda) d\lambda\\=\mathcal O(T^{-2/3}).\end{multline}
\end{lemma}
\begin{proof}
To prove the estimate, we need some preliminary considerations on the Airy point process.
It follows from the general theory of determinantal point processes, see e.g.\ \cite{Soshnikov}, that 
\[\rho(r):=\int_{r}^{+\infty}K^{\Ai}(u,u)du\] is equal to the average number of particles bigger than $r$ in the Airy point process. It was shown in 
\cite[page 5]{CharlierClaeys} that this average has asymptotics
\[\rho(r)=\frac{2}{3\pi}|r|^{3/2}+\mathcal O(|r|^{-3/2}\log |r|)\qquad \mbox{as $r\to -\infty$}.\]
In other words, if we
write
\[F(r)=\rho(r)-\frac{2}{3\pi}|r|^{3/2}{\bf 1}_{(-\infty,0]}(r), \qquad f(r)=-F'(r)=K^{\Ai}(r,r)-
\frac{1}{\pi}|r|^{1/2}{\bf 1}_{(-\infty,0]}(r),
\]
we have that 
\[F(r)=\mathcal O(|r|^{-3/2}\log |r|)\qquad \mbox{as $r\to -\infty$}.\]
Moreover, it follows easily from the asymptotics of the Airy function at $+\infty$ that $F(r)$ is exponentially small as $r\to +\infty$.
It follows that we can bound $F$ uniformly as follows: there exists a constant $C>0$ such that
\begin{equation}|F(r)|\leq C\min\{|r|^{-5/4}, 1\},\qquad r\in\mathbb R.\label{ineqF}\end{equation}

(the exponent $5/4$ has been chosen for the sake of concreteness; the equation above is still valid replacing $5/4$ with any positive constant less than $3/2$). Now, we proceed with estimating the left hand side of \eqref{eq:boundAiry}.
For small enough $\epsilon>0$, we know from Proposition \ref{prop:asendpoint} that $\lambda_0-\epsilon$ is bounded below by a positive constant independent of $s$ and $T$, hence $\sigma(sT^{1/3}\lambda)-1$ is $\mathcal O({\rm e}^{-csT^{1/3}})$ for some $c>0$ as $s\to\infty$, uniformly in $T$ and in $\lambda\in[\lambda_0-\epsilon, \lambda_0+\epsilon]$.
It follows that we can estimate the left hand side of \eqref{eq:boundAiry} as
\[-\frac{s^2T^{-2/3}}{3}
\int_{\lambda_0-\epsilon}^{\lambda_0+\epsilon}
\sigma(sT^{1/3}\lambda)f(s\mu(\lambda))\lambda\mu'(\lambda)d\lambda=
-\frac{s^2T^{-2/3}}{3}
\int_{\lambda_0-\epsilon}^{\lambda_0+\epsilon}
f(s\mu(\lambda))\lambda\mu'(\lambda)d\lambda + \mathcal O({\rm e}^{-\frac{c}{2}sT^{1/3}}).\]
Furthermore, integrating by parts, we obtain
\begin{multline*}
-\frac{s^2T^{-2/3}}{3}\int_{\lambda_0-\epsilon}^{\lambda_0+\epsilon}
f(s\mu(\lambda))\lambda\mu'(\lambda)d\lambda=
\frac{sT^{-2/3}}{3}(\lambda_0+\epsilon) F(s\mu(\lambda_0+\epsilon))\\-
\frac{sT^{-2/3}}{3}(\lambda_0-\epsilon) F(s\mu(\lambda_0-\epsilon))
-\frac{sT^{-2/3}}{3}\int_{\lambda_0-\epsilon}^{\lambda_0+\epsilon}
F(s\mu(\lambda))d\lambda.
\end{multline*}
The integrated terms above are $\mathcal O(s^{-1/4}T^{-2/3})$ by \eqref{ineqF}. For the remaining term, we again use \eqref{ineqF} and obtain
\begin{align*}
\left|-\frac{sT^{-2/3}}{3}\int_{\lambda_0-\epsilon}^{\lambda_0+\epsilon}
F(s\mu(\lambda))d\lambda\right|
&\leq \frac{CsT^{-2/3}}{3}\int_{\lambda_0-\epsilon}^{\lambda_0+\epsilon} \min\{|s\mu(\lambda)|^{-5/4}, 1\}d\lambda\\
&\leq \frac{CsT^{-2/3}}{3}\int_{\lambda_0-\epsilon}^{\lambda_0+\epsilon} \min\{|cs(\lambda-\lambda_0)|^{-5/4}, 1\}d\lambda=\mathcal O(T^{-2/3}),
\end{align*}
uniformly in $T$ as $s\to\infty$,
where we used in the last line that there is a constant $c>0$ such $\mu(\lambda)\geq c|\lambda-\lambda_0|$ for $\lambda\in[\lambda_0-\epsilon, \lambda_0+\epsilon]$ (this can be seen from \eqref{eq:boundmu0} and \eqref{def:mu}). 
%
%
%
\end{proof}

\begin{lemma}\label{lemma:mainterm}
As $s\to\infty$, we have uniformly in $M^{-1}\leq T\leq Ms^{3/2}$ that
\begin{multline}\label{eq:intdensity}
-\frac{s^{5/2}}{6\pi T^{2/3} }\int_{-\infty}^{\lambda_0} \lambda\sigma(sT^{1/3})\psi(\lambda) d\lambda=-\frac{4T}{45\pi^6}\left(\sqrt{1+\pi^2sT^{-2/3}}-1\right)^5
\\{-\frac{T}{9 \pi^6}}\left(\sqrt{1+\pi^2sT^{-2/3}}-1\right)^4
+{\frac{\sqrt{1+\pi^2sT^{-2/3}}}{18 T}}
+\mathcal O(T^{-1}\log s).
\end{multline}
\end{lemma}
\begin{proof}
We write
\begin{align*}-\frac{s^{5/2}}{6\pi T^{2/3} }\int_{-\infty}^{\lambda_0}\lambda\sigma(sT^{1/3})\psi(\lambda) d\lambda&=-\frac{s^{5/2}}{6\pi T^{2/3} }\int_0^{\lambda_0} \lambda \psi(\lambda) d\lambda\\
&-\frac{s^{5/2}}{6\pi T^{2/3} }\int_{-\infty}^{\lambda_0}\lambda\left(\sigma(sT^{1/3}\lambda)-{\bf 1}_{(0,+\infty)}(\lambda)\right)\psi(0) d\lambda\\
&-\frac{s^{5/2}}{6\pi T^{2/3} }\int_{-\infty}^{\lambda_0}\lambda\left(\sigma(sT^{1/3}\lambda)-{\bf 1}_{(0,+\infty)}(\lambda)\right)(\psi(\lambda)-\psi(0)) d\lambda.
\end{align*}
The last line above can be estimated using \eqref{eq:psiest2} and \eqref{ineq:sigma}: it is bounded in absolute value (uniformly in $T$ as $s\to\infty$) by
\[\frac{s^{5/2}}{6\pi T^{2/3}}
\int_{-\infty}^{\lambda_0}
|\lambda| {\rm e}^{-sT^{1/3}|\lambda|}\left|\mathcal O(\lambda)+\mathcal O(s^{-1/2}T^{1/3}\log(2+sT^{1/3}|\lambda|))\right| d\lambda=\mathcal O\left(s^{-1/2}T^{-5/3}\right)
+\mathcal O(T^{-1}).
\]
Next, we use \eqref{eq:psiest1} to conclude that
\begin{align*}
&-\frac{s^{5/2}}{6\pi T^{2/3} }\int_{-\infty}^{\lambda_0}\lambda\left(\sigma(sT^{1/3}\lambda)-{\bf 1}_{(0,+\infty)}(\lambda)\right)\psi(0) d\lambda\\
&\qquad =
\left(-\frac{s^{5/2}\sqrt{\lambda_0}}{3\pi T^{2/3} }-\frac{s^2\log(sT^{1/3})}{6\pi^2 T^{1/3}}+\mathcal O(s^2T^{-1/3})\right)\int_{-\infty}^{\lambda_0}\lambda\left(\sigma(sT^{1/3}\lambda)-{\bf 1}_{(0,+\infty)}(\lambda)\right)d\lambda\\
&\qquad =
\left(-\frac{s^{1/2}\sqrt{\lambda_0}}{3\pi T^{4/3} }-\frac{\log(sT^{1/3})}{6\pi^2 T}+\mathcal O(T^{-1})\right)\int_{-\infty}^{sT^{1/3}\lambda_0}u\left(\sigma(u)-{\bf 1}_{(0,+\infty)}(u)\right) du\\
&\qquad ={\frac{\pi^2}6}\left(\frac{s^{1/2}\sqrt{\lambda_0}}{3\pi T^{4/3} }+\frac{\log(sT^{1/3})}{6\pi^2 T}\right) +\mathcal O(T^{-1})
.
\end{align*}

\medskip

It remains to compute
$-\ds\frac{s^{5/2}}{6\pi T^{2/3} }\int_0^{\lambda_0}\lambda \psi(\lambda) d\lambda$.
By \eqref{exprpsi}, it is straightforward to show that
\[-\frac{s^{5/2}}{6\pi T^{2/3} }\int_0^{\lambda_0}\lambda\psi(\lambda) d\lambda=-\frac{4s^{5/2}}{45\pi T^{2/3}}\lambda_0^{5/2}-\frac{s^2}{6\pi^2 T^{1/3}} \int_{-\infty}^{\lambda_0}\frac{\sigma(sT^{1/3}\xi)}{\sqrt{\lambda_0-\xi}}\left(P.V.\int_0^{\lambda_0}\lambda\frac{\sqrt{\lambda_0-\lambda}}{\xi-\lambda} d\lambda\right)d\xi.\]
The Cauchy principal value integral appearing in this expression can be computed explicitly as
\[P.V.\int_0^{\lambda_0}\lambda\frac{\sqrt{\lambda_0-\lambda}}{\xi-\lambda} d\lambda=-\frac{2}{3}\lambda_0^{3/2}+2\sqrt{\lambda_0}\xi
-2\sqrt{\lambda_0-\xi}\xi
\log\left(\sqrt{\lambda_0-\xi}+\sqrt{\lambda_0}\right)
+\sqrt{\lambda_0-\xi}\xi\log \xi.\]
Substituting this, we obtain
\begin{multline*}-\frac{s^{5/2}}{6\pi T^{2/3} }\int_0^{\lambda_0}\lambda\psi(\lambda) d\lambda=-\frac{4s^{5/2}}{45\pi T^{2/3}}\lambda_0^{5/2}+\frac{s^2}{3\pi^2 T^{1/3}}\int_{-\infty}^{\lambda_0}\left(\frac{\lambda_0^{3/2}}{3}-\sqrt{\lambda_0}\xi\right)\frac{\sigma(sT^{1/3}\xi)}{\sqrt{\lambda_0-\xi}}
 d\xi\\
+\frac{s^2}{3\pi^2 T^{1/3}}
 \int_{-\infty}^{\lambda_0}\sigma(sT^{1/3}\xi)\xi \log\left(\frac{\sqrt{\lambda_0-\xi}+\sqrt{\lambda_0}}{\sqrt{\xi}}\right)
d\xi
.\end{multline*}
We can now use Lemma \ref{lemma:stepapprox} (part 1 for the first integral, part 2 for the second) to conclude that
\begin{multline*}-\frac{s^{5/2}}{6\pi T^{2/3} }\int_0^{\lambda_0}
\lambda\psi(\lambda) d\lambda=-\frac{4s^{5/2}}{45\pi T^{2/3}}\lambda_0^{5/2}+\frac{s^2}{3\pi^2 T^{1/3}}\int_{0}^{\lambda_0}
\frac{\frac{\lambda_0^{3/2}}{3}-\sqrt{\lambda_0}\xi}{\sqrt{\lambda_0-\xi}}d\xi\\
+\frac{s^2}{3\pi^2 T^{1/3}}
 \int_{0}^{\lambda_0}\xi \log\left(\frac{\sqrt{\lambda_0-\xi}+\sqrt{\lambda_0}}{\sqrt{\xi}}\right)
d\xi +\mathcal O(T^{-1}\log s)\end{multline*}
uniformly in $T$ as $s\to\infty$.

Computing the remaining integrals explicitly, we find after a straightforward calculation,
\[
{-\ds\frac{s^{5/2}}{6\pi T^{2/3} }\int_0^{\lambda_0}\lambda \psi(\lambda) d\lambda}=-\frac{4s^{5/2}}{45\pi T^{2/3}}\lambda_0^{5/2}
-\frac{s^2\lambda_0^2}{9\pi^2 T^{1/3}}
+\mathcal O(T^{-1}\log s)
\]
uniformly in $T$ as $s\to \infty$.
Now we use Proposition \ref{prop:asendpoint}, which yields the result after using the fact that $T\leq Ms^{3/2}$.
\end{proof}

Combining Proposition \ref{prop:asendpoint}, Proposition \ref{prop:logder}, Lemma \ref{lemma:oscillint}, Lemma \ref{lemma:Airyterm} and Lemma \ref{lemma:mainterm},
we obtain the main result of this subsection.
\begin{proposition}\label{prop:logder2}
We have
\begin{multline}
\partial_T\log Q(s,T)=
-\frac{4T}{45\pi^6}\left(\sqrt{1+\pi^2sT^{-2/3}}-1\right)^5
{-\frac{T}{9 \pi^6}}\left(\sqrt{1+\pi^2sT^{-2/3}}-1\right)^4\\
+{ \frac{\sqrt{1+\pi^2sT^{-2/3}}}{18 T}}
+\mathcal O(T^{-2/3})+ \mathcal O(T^{-1}\log s)\label{diffidtotal2}
\end{multline}
as $s\to\infty$, uniformly in $T$.
\end{proposition}

\subsection{Logarithmic derivative with respect to $s$}

We need similar asymptotics for the logarithmic $s$-derivative of $Q(s,T)$, but we will only need those for fixed $T>0$, which simplifies the computations considerably. The analogue of Proposition \ref{prop:logder} is the following result.
\begin{proposition}\label{prop:logders}
As $s\to\infty$, we have the following asymptotics for any $T>0$:
\begin{align}
&\partial_s\log Q(s,T)\nonumber\\
&\quad =
-\frac{s^{3/2}T^{1/3}}{2\pi }\int_{-\infty}^{\lambda_0}\sigma(sT^{1/3}\lambda)
\psi(\lambda)d\lambda 
+\frac{T^{1/3}}{4\pi}\int_{-\infty}^{\lambda_0-\epsilon}\frac{\cos\left(
s^{3/2}\int_{\lambda_0}^\lambda\psi(\xi)d\xi\right)}{\lambda-\lambda_0}\sigma(sT^{1/3}\lambda)d\lambda
\nonumber\\ 
&\qquad -sT^{1/3}
\int_{\lambda_0-\epsilon}^{\lambda_0+\epsilon}
\sigma(sT^{1/3}\lambda)
\left(K^{\Ai}(s\mu(\lambda),s\mu(\lambda))-
\frac{1}{\pi}s^{1/2}|\mu(\lambda)|^{1/2}{\bf 1}_{(-\infty,\lambda_0]}(\lambda)\right) 
\mu'(\lambda) d\lambda \nonumber \\ 
&\qquad +\mathcal O(s^{-1}).
\label{diffidtotals}
\end{align}
\end{proposition}
\begin{proof}
We can write \eqref{diffidentityeq_s} in terms of $\widehat S$ as follows,
\[\partial_s\log Q(s,T)=-\frac{T^{1/3}}{2\pi i}\int_{\mathbb R}\sigma(sT^{1/3}\lambda)\left(\widehat S^{-1}\widehat S'\right)_{21}(\lambda){\rm e}^{-s^{3/2}(2g_+(\lambda)-V(\lambda){+V(\lambda_0)})}d\lambda.\]
As before, we split this integral in $3$ parts,
\[\partial_T\log Q(s,T)=\widehat I_1+\widehat I_2+\widehat I_3,\]
where
\begin{equation}\label{def:hatIj}\widehat I_j=-
\frac{T^{1/3}}{2\pi i}\int_{A_j}\sigma(sT^{1/3}\lambda)
\left(\widehat S^{-1}\widehat S'\right)_{21}(\lambda){\rm e}^{-s^{3/2}(2g_+(\lambda)-V(\lambda){+V(\lambda_0)})}d\lambda.
\end{equation}

In a similar way as for $I_3$, we obtain using \eqref{asS1} {and \eqref{ineq1}} that $\widehat I_3
=\mathcal O(s^{-1/2}),$ as $s\to\infty$. 
For $\widehat I_1$, similarly as for $I_1$, substituting \eqref{asS2} in \eqref{def:hatIj}, we obtain
\begin{multline}\label{hatI1beforelemma}
\widehat I_1=-\frac{s^{3/2}T^{1/3}}{2\pi}\int_{-\infty}^{\lambda_0-\epsilon}\sigma(sT^{1/3}\lambda)
\psi(\lambda)d\lambda 
+\frac{T^{1/3}}{4\pi }\int_{-\infty}^{\lambda_0-\epsilon}
\frac{\cos\left(
s^{3/2}\int_{\lambda_0}^\lambda\psi(\xi)d\xi\right)}{\lambda-\lambda_0}\sigma(sT^{1/3}\lambda)d\lambda
\\
+\mathcal O\left(s^{-3/2}\right),
\end{multline}
as $s\to\infty$.
Finally, for $\widehat I_2$, we obtain using \eqref{asS3} that
\begin{equation}\label{hatI2}
\widehat I_2=-sT^{1/3}\int_{\lambda_0-\epsilon}^{\lambda_0+\epsilon}\sigma(sT^{1/3}\lambda)K^{\Ai}(s\mu(\lambda),s\mu(\lambda))\mu'(\lambda) d\lambda  
+\mathcal O(s^{-1}),
\end{equation}
as $s\to\infty$.

Summing up the above contributions and using the fact that $\psi(\lambda)=-2i\mu(\lambda)^{1/2}\mu'(\lambda)$, we obtain the result.
\end{proof}

The following is the counterpart of Lemma \ref{lemma:oscillint}. We omit the proof as it is similar to the one given in the previous subsection.
\begin{lemma}\label{lemma:oscillints}
As $s\to\infty$, we have for any fixed $T>0$,
\[\frac{T^{1/3}}{4\pi}\int_{-\infty}^{\lambda_0-\epsilon}\frac{\cos\left(
s^{3/2}\int_{\lambda_0}^\lambda\psi(\xi)d\xi\right)}{\lambda-\lambda_0}\sigma(sT^{1/3}\lambda) d\lambda
=\mathcal O(s^{-3/2}).\]
\end{lemma}
Instead of Lemma \ref{lemma:Airyterm}, we have the following.
\begin{lemma}\label{lemma:Airyterms}
As $s\to\infty$, we have for any fixed $T>0$,
\begin{equation}\label{eq:boundAirys}-sT^{1/3}
\int_{\lambda_0-\epsilon}^{\lambda_0+\epsilon}
\sigma(sT^{1/3}\lambda)
\left(K^{\Ai}(s\mu(\lambda),s\mu(\lambda))-
\frac{1}{\pi}s^{1/2}|\mu(\lambda)|^{1/2}{\bf 1}_{(-\infty,\lambda_0]}(\lambda)\right) 
\mu'(\lambda) d\lambda=\mathcal O(s^{-5/4}).\end{equation}
\end{lemma}
\begin{proof}
With the same notations as in the proof of Lemma \ref{lemma:Airyterm}, we need to bound
\[-sT^{1/3}
\int_{\lambda_0-\epsilon}^{\lambda_0+\epsilon}
\sigma(sT^{1/3}\lambda)f(s\mu(\lambda)) 
\mu'(\lambda) d\lambda,\] which is equal to 
$-sT^{1/3}
\int_{\lambda_0-\epsilon}^{\lambda_0+\epsilon}
f(s\mu(\lambda)) 
\mu'(\lambda) d\lambda$
up to exponentially small error terms because $\sigma(sT^{1/3}\lambda)$ is exponentially close $1$.
But 
\begin{multline*}-sT^{1/3}
\int_{\lambda_0-\epsilon}^{\lambda_0+\epsilon}
f(s\mu(\lambda)) 
\mu'(\lambda) d\lambda=-T^{1/3}
\int_{s\mu(\lambda_0-\epsilon)}^{s\mu(\lambda_0+\epsilon)}
f(u)du\\=T^{1/3}F(s\mu(\lambda_0+\epsilon)) - T^{1/3}F(s\mu(\lambda_0-\epsilon))
=\mathcal O(s^{-5/4})\end{multline*}
as $s\to\infty$, by \eqref{ineqF}.
\end{proof}

Next, we derive the analogue of Lemma \ref{lemma:mainterm}.
\begin{lemma}\label{lemma:mainterms}
As $s\to\infty$, we have for any fixed $T>0$,
\begin{multline*}
-\frac{s^{3/2}T^{1/3}}{2\pi }\int_{-\infty}^{\lambda_0}\sigma(sT^{1/3})\psi(\lambda) d\lambda=-\frac{2T^{4/3}}{3\pi^4 }\left(\sqrt{1+\pi^2sT^{-2/3}}-1\right)^3\\
{-\frac{T^{4/3}}{\pi^4}\left(\sqrt{1+\pi^2sT^{-2/3}}-1\right)^2 }
{-\frac{\pi}{12 s^{1/2}T^{1/3}}}
+\mathcal O(s^{-1}\log s).
\end{multline*}
\end{lemma}
\begin{proof}
The proof follows the same lines as that of Lemma \ref{lemma:mainterm}: we start by observing that
\begin{multline*}-\frac{s^{3/2}T^{1/3}}{2\pi }\int_{-\infty}^{\lambda_0}\sigma(sT^{1/3})\psi(\lambda) d\lambda=
-\frac{s^{3/2}T^{1/3}\psi(0)}{2\pi}\int_{-\infty}^{\lambda_0}\left(\sigma(sT^{1/3}\lambda)-{\bf 1}_{(0,+\infty)}(\lambda)\right) d\lambda\\
-\frac{s^{3/2}T^{1/3}}{2\pi}\int_{-\infty}^{\lambda_0}\left(\sigma(sT^{1/3}\lambda)-{\bf 1}_{(0,+\infty)}(\lambda)\right)\Big(\psi(\lambda)-\psi(0)\Big) d\lambda
-\frac{s^{3/2}T^{1/3}}{2\pi }\int_0^{\lambda_0}\psi(\lambda)d\lambda
.
\end{multline*}
The first term at the right hand side {is exponentially small, as it can be verified remarking that $\int_{-\infty}^\infty (\sigma(u)-{\bf 1}_{(0,+\infty)}(u)) du = 0$ and using \eqref{ineq:sigma}}. The second summand can be estimated by \eqref{ineq:sigma} and \eqref{eq:psiest2} as
\begin{align*}
&-\frac{s^{3/2}T^{1/3}}{2\pi}\int_{-\infty}^{\lambda_0}\left(\sigma(sT^{1/3}\lambda)-{\bf 1}_{(0,+\infty)}(\lambda)\right)(\psi(\lambda)-\psi(0)) d\lambda\\
&=
\frac{s^{3/2}T^{1/3}}{2\pi}\int_{-\infty}^{\lambda_0}\left(\sigma(sT^{1/3}\lambda)-{\bf 1}_{(0,+\infty)}(\lambda)\right)\\
&\qquad \times \ \left(
2\sqrt{\lambda_0-\lambda}-2\sqrt{\lambda_0}+\mathcal O(s^{-1/2}T^{1/3}\lambda)+\mathcal O_{\rm sym}\left(s^{-1/2}T^{1/3}\log\left(2+ sT^{1/3}|\lambda|)\right)\right)\right) d\lambda \\
&={-\frac{\pi}{12 s^{1/2}T^{1/3}}}+\mathcal O(s^{-1})
\end{align*}
as $s\to\infty$ for fixed $T$.
For this, we used the antisymmetry of $\sigma(sT^{1/3}\lambda)-{\bf 1}_{(0,+\infty)}(\lambda)$ and the symmetry of the last error term, and the fact that 
 $\lambda_0=1+\mathcal O(s^{-1/2}\log s)$ as $s\to\infty$ for fixed $T$.
It remains to compute
$\ds-\frac{s^{3/2}T^{1/3}}{2\pi }\int_0^{\lambda_0}\psi(\lambda) d\lambda$.
By \eqref{exprpsi}, we have
\[-\frac{s^{3/2}T^{1/3}}{2\pi}\int_0^{\lambda_0}\psi(\lambda) d\lambda=-\frac{2s^{3/2}T^{1/3}}{3\pi }\lambda_0^{3/2}-\frac{sT^{2/3}}{2\pi^2} \int_{-\infty}^{\lambda_0}\frac{\sigma(sT^{1/3}\xi)}{\sqrt{\lambda_0-\xi}}\left(P.V.\int_0^{\lambda_0}\frac{\sqrt{\lambda_0-\lambda}}{\xi-\lambda} d\lambda\right)d\xi.\]
The Cauchy principal value integral appearing in this expression can be computed explicitly as
\[P.V.\int_0^{\lambda_0}\frac{\sqrt{\lambda_0-\lambda}}{\xi-\lambda}d\lambda=2\sqrt{\lambda_0}
-2\sqrt{\lambda_0-\xi}
\log\left(\sqrt{\lambda_0-\xi}+\sqrt{\lambda_0}\right)
+\sqrt{\lambda_0-\xi}\log \xi.\]
Substituting this, we obtain
\begin{multline*}-\frac{s^{3/2}T^{1/3}}{2\pi}\int_0^{\lambda_0}\psi(\lambda) d\lambda=-\frac{2s^{3/2}T^{1/3}}{3\pi }\lambda_0^{3/2}-\frac{sT^{2/3}\sqrt{\lambda_0}}{\pi^2}\int_{-\infty}^{\lambda_0}
\frac{\sigma(sT^{1/3}\xi)}{\sqrt{\lambda_0-\xi}}
d\xi\\
+\frac{sT^{2/3}}{\pi^2}
 \int_{-\infty}^{\lambda_0}\sigma(sT^{1/3}\xi)
 \log\left(\frac{\sqrt{\lambda_0-\xi}+\sqrt{\lambda_0}}{\sqrt{\xi}}\right)
d\xi.
\end{multline*}
Using Lemma \ref{lemma:stepapprox} as before and then computing the remaining integrals explicitly, we get
\[-\frac{s^{3/2}T^{1/3}}{2\pi}\int_0^{\lambda_0}\psi(\lambda) d\lambda=-\frac{2s^{3/2}T^{1/3}}{3\pi }\lambda_0^{3/2}
-\frac{sT^{2/3}\lambda_0}{\pi^2}
+\mathcal O(s^{-1}\log s)\]
as $s\to \infty$.
Now we use Proposition \ref{prop:asendpoint}, which easily yields the result.
\end{proof}

Summing up the different contributions, we arrive at the main result of this subsection.
\begin{proposition}\label{prop:logders2}
As $s\to\infty$, we have the following asymptotics for any fixed $T>0$:
\begin{multline}
\partial_s\log Q(s,T)=
-\frac{2T^{4/3}}{3\pi^4 }\left(\sqrt{1+\pi^2sT^{-2/3}}-1\right)^3
{-\frac{T^{4/3}}{\pi^4 }\left(\sqrt{1+\pi^2sT^{-2/3}}-1\right)^2}\\
{-\frac{s^{-1/2}T^{-1/3}\pi}{12}}
+\mathcal O(s^{-1}\log s).\label{diffidtotal2b}
\end{multline}
\end{proposition}

\subsection{Integration of the differential identity}

In order to derive asymptotics for $\log Q(s,T)$ as $s\to\infty$ which are uniform in $M^{-1}\leq T\leq Ms^{3/2}$, we integrate the asymptotics for the logarithmic derivatives obtained in the previous section. 
To do this, we note that 
\begin{equation}\label{integral}
\log Q(s_0,T_0)=\log Q(M',M^{-1})
+\int_{M^{-1}}^{T_0}\partial_T\log Q(s_0,T)dT+
\int_{M'}^{s_0}\partial_s\log Q(s,M^{-1})ds.
\end{equation}
The first term is independent of $s$ and $T$, and 
if we take $M'$ large enough, $s$ is large such that we can substitute the asymptotics for $\partial_T\log Q(s_0,T)$ (since $M^{-1}\leq T_0\leq Ms_0^{3/2}$ implies that $M^{-1}\leq T\leq Ms_0^{3/2}$ for any $T\in[M^{-1},T_0]$), and the asymptotics for $\partial_s\log Q(s,M^{-1})$ (since $M^{-1}\leq Ms^{3/2}$ for any $s\in[M',s_0]$ with $M'$ large enough).

Recall $\phi(y)$ defined in \eqref{def:phi}. It is straightforward to verify that
\begin{equation}
-\partial_T\left(T^2\phi(sT^{-2/3})\right)=
-\frac{4T}{45\pi^6}\left(\sqrt{1+\pi^2sT^{-2/3}}-1\right)^5
{-\frac{T}{9\pi^6 }\left(\sqrt{1+\pi^2sT^{-2/3}}-1\right)^4}
\end{equation}
and that
\begin{equation}
-\partial_s\left(T^2\phi(sT^{-2/3})\right)=
-\frac{2T^{4/3}}{3\pi^4 }\left(\sqrt{1+\pi^2sT^{-2/3}}-1\right)^3\\
{-\frac{T^{4/3}}{\pi^4 }\left(\sqrt{1+\pi^2sT^{-2/3}}-1\right)^2.}
\end{equation}
By Proposition \ref{prop:logder2} and Proposition \ref{prop:logders2}, we have
\begin{align*}
&\partial_T\log Q(s,T)=-\partial_T T^2\phi(sT^{-2/3})+ { \frac{\sqrt{1+\pi^2sT^{-2/3}}}{18 T}}
+\mathcal O(T^{-1}\log s)+\mathcal O(T^{-2/3}),\\
&\partial_s\log Q(s,T)=-\partial_s T^2\phi(sT^{-2/3}){-\frac{s^{-1/2}T^{-1/3}\pi}{12}}+\mathcal O(s^{-1}\log s),
\end{align*}
as $s\to\infty$ uniformly in $T$.
Using this, we obtain that
\begin{multline}
\int_{M^{-1}}^{T_0}\partial_T\log Q(s_0,T)dT=-T^2_0\phi(s_0T_0^{-3/2})+M^{-2}\phi(s_0M^{3/2})\\
{-\frac{1}{6}}\sqrt{1+\frac{\pi^2s_0}{T_0^{2/3}}}+{\frac{1}{6}}\sqrt{1+\frac{\pi^2s_0}{M^{-2/3}}}
+\mathcal O(\log s_0 \log T_0)+\mathcal O(T_0^{1/3}),
\end{multline}
and that
\begin{equation}
\int_{M'}^{s_0}\partial_s\log Q(s,M^{-1})ds=-M^{-2}\phi(s_0M^{3/2})+M^{-2}\phi(M'M^{2/3})
{-\frac{\pi\sqrt{s_0}}{6 M^{-1/3}}} +\mathcal O(\log^2 s_0).
\end{equation}
Substituting this into \eqref{integral}, we get
\begin{equation}
\log Q(s_0,T_0)=-T_0^2\phi(s_0T_0^{-3/2}){-\frac{1}6}\sqrt{1+\frac{\pi^2s_0}{T_0^{2/3}}} +\mathcal O(\log^2 s_0)+\mathcal O(T_0^{1/3})
\end{equation}
as $s_0\to\infty$, uniformly in $M^{-1}\leq T_0\leq Ms_0^{3/2}$. This completes the proof of Theorem \ref{theorem:main}.

\paragraph{Acknowledgements.}\hfill{}\\
M.C. was supported by the European Union Horizon 2020 research and innovation program under the Marie Sk\l odowska-Curie RISE 2017 grant agreement no. 778010 IPaDEGAN. 
T.C. was supported by the Fonds de la Recherche Scientifique-FNRS under EOS project O013018F. The authors are also grateful to D. Betea, T. Bothner and J. Bouttier for fruitful discussions.

\end{document}